\documentclass[journal]{IEEEtran}

\usepackage[T1]{fontenc}
\usepackage{amsmath,amssymb,mathtools,bm,amsthm}
\usepackage{graphicx}
\usepackage{booktabs,multirow,array}
\usepackage{cite}
\usepackage{xcolor}
\usepackage{url}
\usepackage{tikz}
\usetikzlibrary{arrows.meta,positioning,calc,fit,shapes.geometric,shapes.symbols,decorations.pathreplacing}

\makeatletter
\renewcommand{\@IEEEfigurecaptionsepspace}{\vskip4pt\nointerlineskip}
\makeatother
\definecolor{gridblue}{RGB}{35,103,170}
\definecolor{supportgreen}{RGB}{34,139,94}
\definecolor{tradeoffred}{RGB}{185,45,45}
\definecolor{lightblue}{RGB}{231,241,250}
\definecolor{lightgreen}{RGB}{231,247,238}
\definecolor{lightred}{RGB}{252,237,237}

\newtheorem{theorem}{Theorem}
\newtheorem{proposition}{Proposition}

\newcommand{\RHinf}{\mathcal{RH}_{\infty}}
\newcommand{\Ppcc}{P_{\mathrm{PCC}}}
\newcommand{\Pdc}{P_{\mathrm{dc}}}
\newcommand{\Pb}{P_{\mathrm{BESS}}}

\begin{document}

\title{\huge Data Center Oscillation Isolation and Grid Support:\\
Fundamental Limits and Source-Selective Control}

\author{Cheng Feng,~\IEEEmembership{Member,~IEEE}}

\markboth{}%
{Feng: Data Center Oscillation Isolation and Grid Support}

\maketitle
\setlength{\abovedisplayskip}{0.7ex plus 2pt minus 1pt}
\setlength{\belowdisplayskip}{0.7ex plus 2pt minus 1pt}
\setlength{\belowdisplayshortskip}{0.7ex plus 2pt minus 1pt}

\begin{abstract}
The grid integration of large data centers poses a growing power-system stability challenge. Rapidly varying computing demand can induce oscillations, while the facilities must provide a prescribed response to grid disturbances. Existing controls address load fluctuations and grid support but do not establish their jointly achievable performance. This paper derives fundamental limits and develops a source-selective controller for simultaneous oscillation isolation and grid support. A coupled closed-loop model gives a condition for distinguishing disturbance sources and reveals a direct isolation-support conflict under grid-import-power feedback alone; additional electrical measurements can remove this pointwise conflict but leave a conditional cross-frequency sensitivity limit. Guided by these limits, we design a local residual-feedback storage controller through convex frequency-domain optimization with stability, isolation, and grid-response constraints across operating points. Tests on a single-generator interconnection and a modified nine-bus network show stronger oscillation isolation than equal-order optimized comparators while meeting the prescribed grid response across operating points.
\end{abstract}


\section{Introduction}
\label{sec:introduction}

The grid integration of large-scale data centers has become a significant power-system stability concern. Rapidly varying computing workloads can produce sustained power oscillations, while converter-interfaced equipment can respond abruptly to grid disturbances. The North American Electric Reliability Corporation (NERC) documented a 2023 event in which periodic 1-Hz perturbations at a data center repeatedly excited an 11-Hz electrical ringdown~\cite{nerc2025large_loads}. Following a 230-kV transmission fault in July 2024, data center load disconnected from the grid, raising system voltage and frequency~\cite{nerc2025incident}. These events call for two concurrent capabilities: attenuating workload-induced oscillations and providing a specified response when a disturbance originates in the grid.

Considerable effort has been devoted to characterizing related data center dynamics recently. Ref.~\cite{mishra2025inception} analyzes field measurements of a 14.7-Hz data-center oscillation, while Refs.~\cite{sun2022psu,sun2022system} show how power-supply impedance and facility distribution networks affect resonance. Moving from the facility to the power-system scale, Ref.~\cite{jimenez2025transient} includes UPS, cooling, and pulsed computing loads in a transient-stability model, and Ref.~\cite{chaudhary2026correlated} examines correlated demand across sites. Refs.~\cite{ko2025widearea,zhao2026datacenter,krajacic2026forced} further study workload forcing, converter interactions, and propagation into the interconnected grid.

Beyond modelling, control studies have first treated data centers as flexible resources for grid-frequency support. Ref.~\cite{li2013regulation} modulates server power to track a regulation signal, and Ref.~\cite{guruprasad2017regulation} coordinates workload scheduling with battery power. A UPS field pilot demonstrates fast frequency response~\cite{alapera2019ups}, while coordinated workload-UPS control extends this capability to frequency events~\cite{tao2025frequency}. These studies establish that data centers can provide active-power support, but their control objectives do not explicitly bound the oscillatory power exported to the grid by routine computing workloads. A separate line of work focuses on mitigating data-center-induced power fluctuations. Refs.~\cite{yang2018tieline,yang2024multicenter} coordinate workload scheduling and UPS storage to smooth tie-line power. Ref.~\cite{you2026hess} uses hybrid storage to compensate AI-workload variations, and Ref.~\cite{chen2026modeaware} targets frequency components near vulnerable grid modes. Although effective against load-origin disturbances, smoothing alone does not specify the data center response to external grid disturbance.

Recent storage-based work considers both load variation and grid events. Ref.~\cite{lyu2026pcc} reports reduced data-center power fluctuations together with improved voltage and frequency dynamics across grid strengths. Ref.~\cite{abudyak2026gfm} combines load-transient compensation with voltage support, Ref.~\cite{ross2026gfm} examines fast load fluctuations and low-voltage faults, and Ref.~\cite{azizi2026ups} combines power smoothing with grid-connected and islanded operation. These studies demonstrate multiple useful functions, but do not formulate attenuation of workload-induced grid-import fluctuations and response to grid disturbances as two simultaneous closed-loop specifications. Grid and server disturbances jointly affect the powers measured at the grid connection and within the data center, so those measurements cannot automatically be treated as separate disturbance inputs. Moreover, demonstrating both functions in selected cases does not establish which responses are simultaneously achievable across frequencies and operating points. This motivates an analysis of source distinguishability and closed-loop performance limits before controller design.

This paper addresses these questions using local energy storage to isolate server-demand oscillations while tracking a prescribed grid response. The main contributions are:
\begin{enumerate}[\IEEEnocalcleftmargintrue\setlength{\leftmargin}{2pt}\setlength{\itemindent}{1.2em}\setlength{\labelwidth}{0.7em}\setlength{\labelsep}{0.5em}]
\item Develop a coupled data-center--grid port model and derive a rank condition for disturbance-source identifiability. The condition explains why measuring both grid-import and data-center powers does not, by itself, separate grid- and server-origin disturbances in the closed loop.
\item Derive a pointwise performance bound that exposes the conflict between oscillation isolation and grid-response tracking under grid-import-power feedback alone. With additional electrical measurements, identify when a cross-frequency sensitivity limit remains despite improved source discrimination.
\item Design a local residual-feedback storage controller through a convex frequency-domain program with stability, isolation, and grid-response constraints. Equal-order comparisons in single-generator and modified IEEE nine-bus systems evaluate independently designed controllers across operating points.
\end{enumerate}

The remainder of this paper is organized as follows. Section~\ref{sec:model} develops the data-center--grid model and closed-loop transfers. Section~\ref{sec:limits} derives the performance limitations, Section~\ref{sec:controller} presents the controller design, and Section~\ref{sec:case_studies} evaluates frequency-domain and time-domain performance across operating points.

\section{System Model and Closed-Loop Control Framework}
\label{sec:model}

\subsection{Physical Architecture and Small-Signal Port Model}

Fig.~\ref{fig:physical} shows the considered grid-connected AI data center. A hybrid energy storage system (HESS), consisting of a battery and a supercapacitor, is connected at the point of common coupling (PCC). In addition to isolating server-demand-induced oscillations from the grid, the HESS is required to preserve or deliberately shape the response of the data center to grid-side events, such as frequency deviations, voltage disturbances, and poorly damped modes.

Define the PCC grid-port power as $p=\Delta\Ppcc$, realized HESS power as $u=\Delta\Pb$, and data-center power as $d=\Delta\Pdc$. The local electrical variables are the PCC voltage-magnitude deviation $v=\Delta V_{\rm PCC}$ and the local frequency deviation $f=\Delta f_{\rm PCC}$ derived from the PCC voltage-phasor angle. Collect these two variables as $\bm{\xi}=[v,f]^{\mathsf T}$. The input $g$ represents a disturbance at a specified location in the external grid. In Figs.~\ref{fig:physical} and~\ref{fig:framework}, $g=\Delta P_{\rm ext}$ is an increase in demand at a remote bus. The other external input is the server-demand perturbation $w=\Delta P_{\rm dem}$, generated by changes in computing activity. At an operating point, the data-center--grid dynamics are
\begin{equation}
 \begin{aligned}
  d&=D_w(s)w+\bm{D}_\xi^{\mathsf T}(s)\bm{\xi},\quad \bm{D}_\xi=[D_v\ D_f]^{\mathsf T},\\[-0.3mm]
  \bm{\xi}&=\bm{G}_p(s)p+\bm{G}_g(s)g,\qquad p=d-u.
 \end{aligned}
 \label{eq:physical_loop}
\end{equation}
The first equation describes the aggregate dynamics of the servers and their power-conversion chain, represented by $\bm{G}_{\rm dc}=[D_w\ \bm{D}_\xi^{\mathsf T}]$ in Fig.~\ref{fig:physical}. Its inputs are server demand $w$, PCC voltage $v$, and PCC frequency $f$; its output $d$ is the power drawn by the data center. Specifically, $D_w$ describes how server demand changes $d$ with PCC electrical conditions fixed, whereas the entries of $\bm{D}_\xi$ describe the effects of voltage and frequency at a fixed server demand. The second equation describes the external grid, represented by $\bm{G}_{\rm g}=[\bm{G}_p\ \bm{G}_g]$. It takes PCC withdrawal $p$ and the remote disturbance $g$ as inputs and produces the local voltage and frequency deviations $\bm{\xi}=[v,f]^{\mathsf T}$. Here $\bm{G}_p$ gives the response to PCC withdrawal and $\bm{G}_g$ the response to the remote disturbance; their first and second entries describe voltage and frequency, respectively. These transfers are obtained by linearizing the external-grid model and selecting the corresponding input and output channels~\cite{alican2026propagation}. The formulation accommodates general transfer-function models.

At the PCC bus, the grid and HESS jointly supply the data-center demand with $p=d-u$. This balance connects the two physical blocks in Fig.~\ref{fig:framework}: a change in $d$ or $u$ changes $p$, the external grid responds through $\bm{G}_p$, and the resulting voltage and frequency variations feed back to the data center through $\bm{D}_\xi$. An external disturbance enters the same loop through $\bm{G}_g$. Eliminating $d$ and $\bm{\xi}$ from Eq.~\eqref{eq:physical_loop} gives $(1-\bm{D}_\xi^{\mathsf T}\bm{G}_p)p=D_ww+\bm{D}_\xi^{\mathsf T}\bm{G}_gg-u$. Define the scalar return difference $\Delta=1-\bm{D}_\xi^{\mathsf T}\bm{G}_p$. The resulting transfers from server demand $w$, grid disturbance $g$, and HESS power $u$ to $p$ and $\bm y$ are
{\setlength{\abovedisplayskip}{0pt}\setlength{\abovedisplayshortskip}{0pt}%
\par\vspace{-0.45\baselineskip}
\begin{subequations}
	\label{eq:port_model}
	\begin{align}
		p &= P_w(s)w+P_g(s)g+P_u(s)u,\label{eq:port_power}\\
		\bm{y} &= \bm{h}_w(s)w+\bm{h}_g(s)g+\bm{h}_u(s)u,\label{eq:port_measurement}
	\end{align}
\end{subequations}}
where
\begin{equation}
 P_w=\frac{D_w}{\Delta},\quad
 P_g=\frac{\bm{D}_\xi^{\mathsf T}\bm{G}_g}{\Delta},\quad P_u=-\frac{1}{\Delta},
 \label{eq:reduced_port_maps}
\end{equation}
and
\begin{equation}
 \bm{h}_w=\begin{bmatrix}P_w\\\bm{G}_pP_w\end{bmatrix},
 \bm{h}_g=\begin{bmatrix}P_g\\\bm{G}_pP_g+\bm{G}_g\end{bmatrix},
 \bm{h}_u=\begin{bmatrix}P_u\\\bm{G}_pP_u\end{bmatrix}.
 \label{eq:physical_signatures}
\end{equation}
The common denominator $\Delta$ in \eqref{eq:reduced_port_maps} retains the physical grid and data-center feedback.

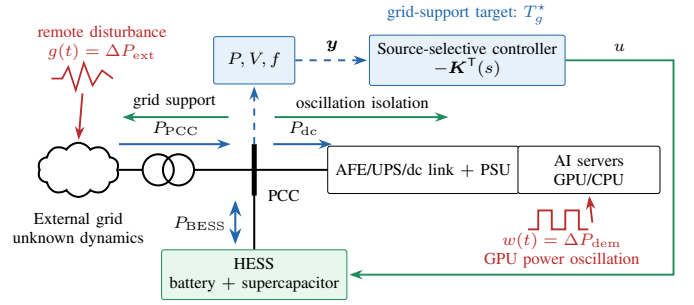
\begin{figure}[!t]
	\centering
	\resizebox{\linewidth}{!}{%
		\begin{tikzpicture}[
			>=Stealth,font=\scriptsize,
			elec/.style={line width=0.85pt},
			bus/.style={line width=2.1pt},
			pflow/.style={-{Stealth[length=2.0mm,width=1.4mm]},draw=gridblue,line width=0.85pt,shorten >=0.8pt},
			sig/.style={-{Stealth[length=2.0mm,width=1.4mm]},draw=supportgreen,line width=0.8pt,shorten >=0.8pt},
			meas/.style={-{Stealth[length=2.0mm,width=1.4mm]},draw=gridblue,dashed,line width=0.8pt,shorten >=0.8pt},
			disturb/.style={-{Stealth[length=2.0mm,width=1.4mm]},draw=tradeoffred,line width=0.8pt,shorten >=0.8pt},
			block/.style={draw,rounded corners=1pt,minimum height=7mm,align=center,fill=white},
			blue/.style={block,draw=gridblue,fill=lightblue},
			green/.style={block,draw=supportgreen,fill=lightgreen}
			]
			\node[draw=black,cloud,cloud puffs=10,cloud puff arc=120,aspect=1.5,minimum width=12mm,minimum height=8mm,inner sep=0pt,line width=0.85pt,fill=none] (grid) at (0.4,0) {};
			\node[below=1mm of grid,align=center] {External grid\\unknown dynamics};
			\draw[elec] (grid.east) -- (1.35,0);
			\draw[elec] (1.58,0) circle (0.23);
			\draw[elec] (1.88,0) circle (0.23);
			\draw[elec] (2.11,0) -- (3.0,0);
			\draw[bus] (3.0,-0.38) -- (3.0,0.38);
			\node[below right=0.1mm and 0.8mm] at (3.0,-0.16) {PCC};
			\coordinate (pccA) at (3.0,0);
			\draw[elec] (pccA) -- (4.25,0);
			\node[block,minimum width=25mm] (train) at (5.50,0)
			{AFE/UPS/dc link $+$ PSU};
			\draw[elec] (train.east) -- (6.90,0);
			\node[block,minimum width=21mm] (server) at (7.95,0)
			{AI servers\\GPU/CPU};
			
			\draw[pflow] (1.0,0.38) -- node[above,text=black] {$\Ppcc$} (2.68,0.38);
			\draw[pflow] (3.28,0.38) -- node[above,text=black] {$\Pdc$} (4.18,0.38);
			\draw[pflow,draw=supportgreen] (2.62,0.78) -- node[above,text=black]
			{grid support} (1.02,0.78);
			\draw[pflow,draw=supportgreen] (3.30,0.78) -- node[above,text=black]
			{oscillation isolation} (5.90,0.78);
			
			\node[green,minimum width=24mm] (hess) at (3.0,-1.55)
			{HESS\\battery $+$ supercapacitor};
			\draw[elec] (pccA) -- (hess.north);
			\draw[pflow,<->] (2.72,-0.45) -- node[left,text=black] {$\Pb$} (2.72,-1.13);
			
			\node[blue,minimum width=28mm] (controller) at (6.15,1.62)
			{Source-selective controller\\$-\bm{K}^{\mathsf T}(s)$};
			\node[blue,minimum width=12mm] (meter) at (3.0,1.62) {$P,V,f$};
			\draw[meas] (3.0,0.38) -- (meter.south);
			\draw[meas] (meter.east) -- node[above,text=black] {$\bm{y}$} (controller.west);
			\coordinate (outerA) at (9.20,1.62);
			\draw[sig] (controller.east) -- node[above,text=black] {$u$} (outerA)
			-- (9.20,-1.55) -- (hess.east);
			\node[text=gridblue] at (6.15,2.30) {grid-support target: $T_g^\star$};
			
			\draw[tradeoffred,line width=0.75pt,-]
			(0.00,1.35) -- (0.16,1.35) -- (0.25,1.52) -- (0.36,1.18)
			-- (0.48,1.58) -- (0.61,1.25) -- (0.75,1.47) -- (0.91,1.35);
			\node[text=tradeoffred,align=center] at (0.75,1.87)
			{remote disturbance\\$g(t)=\Delta P_{\rm ext}$};
			\draw[disturb] (0.46,1.13) -- (grid.north);
			
			\draw[tradeoffred,line width=0.8pt,-]
			(7.05,-0.88) -- (7.20,-0.88) -- (7.20,-0.60) -- (7.40,-0.60)
			-- (7.40,-0.88) -- (7.60,-0.88) -- (7.60,-0.60) -- (7.80,-0.60)
			-- (7.80,-0.88) -- (7.94,-0.88);
			\node[text=tradeoffred,align=center] at (7.55,-1.18)
			{$w(t)=\Delta P_{\rm dem}$\\GPU power oscillation};
			\draw[disturb] (8.00,-0.78) -- (server.south);
		\end{tikzpicture}%
	}
	\caption{Physical architecture. Server demand $w$ originates inside the data center; $g$ is an external-grid disturbance, illustrated by a power imbalance at a remote bus. The cloud represents external-grid dynamics with model uncertainty. The HESS uses local PCC power, voltage, and frequency for oscillation isolation and grid support.}
	\label{fig:physical}
\end{figure}

\begin{figure}[!t]
	\centering
	\begin{tikzpicture}[
		x=1mm,y=1mm,font=\scriptsize,
		>=Stealth,
		wire/.style={line width=0.65pt},
		signal/.style={wire,-{Stealth[length=1.6mm,width=1.15mm]}},
		physical/.style={signal,draw=gridblue},
		measurement/.style={physical,dashed},
		disturbance/.style={signal,draw=tradeoffred},
		control/.style={signal,draw=supportgreen},
		block/.style={draw,wire,rounded corners=1pt,
			minimum width=16mm,minimum height=7mm,
			inner sep=1mm,align=center,fill=white},
		green/.style={block,draw=supportgreen,fill=lightgreen},
		sum/.style={draw,circle,wire,inner sep=0pt,minimum size=4.6mm,fill=white},
		label/.style={inner sep=0.5mm,text=black}
		]
		\node[block,minimum width=18mm,minimum height=11mm] (dc) at (25,30)
		{Data center\\$\bm{G}_{\rm dc}(s)$\\$[D_w\quad \bm{D}_\xi^{\mathsf T}]$};
		\node[sum] (balance) at (44,30) {};
		\node[block,minimum width=18mm,minimum height=11mm] (gridB) at (65,30)
		{External grid\\$\bm{G}_{\rm g}(s)$\\$[\bm{G}_p\quad \bm{G}_g]$};
		\node[green,minimum width=20mm] (feedback) at (44,10)
		{HESS control\\$u=-\bm{K}^{\mathsf T}(s)\bm y$\\$\bm{K}=[K_p\ K_v\ K_f]^{\mathsf T}$};
		
		\draw[disturbance] (0,30) -- node[above=0.7mm,label,text=tradeoffred]
		{$w=\Delta P_{\rm dem}$} (dc.west);
		\draw[disturbance] (86,30) -- node[above=6.8mm,label,text=tradeoffred]
		{$g=\Delta P_{\rm ext}$} (gridB.east);
		
		\draw[physical] (dc.east) -- node[above=6.8mm,label,pos=0.20]
		{$d=\Delta P_{\rm dc}$} (balance.west);
		\node at (40.5,27.8) {$+$};
		\node at (46.4,26.4) {$-$};
		\draw[physical] (balance.east) -- node[above=6.8mm,label,pos=0.85,text=tradeoffred]
		{$p=\Delta P_{\rm PCC}$} (gridB.west);
		\coordinate (pTap) at (50,30);
		\fill[tradeoffred] (pTap) circle (0.45mm);
		
		\coordinate (vTap) at (65,10);
		\draw[wire,draw=gridblue] (gridB.south) -- (vTap);
		\draw[physical] (vTap) -- (65,-1) -- node[below=0.7mm,label]
		{$\bm{\xi}=[v,f]^{\mathsf T}=[\Delta V_{\rm PCC},\Delta f_{\rm PCC}]^{\mathsf T}$} (25,-1) -- (dc.south);
		\fill[gridblue] (vTap) circle (0.45mm);
		
		\draw[measurement] (pTap) -- node[right,label] {$p$}
		([xshift=6mm]feedback.north);
		\draw[measurement] (vTap) -- node[above,label] {$\bm{\xi}$} (feedback.east);
		\draw[control] (feedback.west) -- (30,10) -- (30,22)
		-- node[below=0.5mm,label,pos=0.75] {$u=\Delta P_{\rm BESS}$} (44,22)
		-- (balance.south);
	\end{tikzpicture}
	\caption{Closed-loop block diagram. The data center maps $(w,\bm{\xi})$ to $d$; the external grid maps $(p,g)$ to the local voltage/frequency vector $\bm{\xi}$. The junction imposes $p=d-u$. Dashed branches supply $\bm{y}=[p,v,f]^{\mathsf T}$ to the HESS controller, with $u=-K_pp-K_vv-K_ff$ and actuator dynamics included. Closing the loop gives $T_w:w\to p$ and $T_g:g\to p$; $v$ and $f$ are endogenous outputs of the coupled system.}
	\label{fig:framework}
\end{figure}
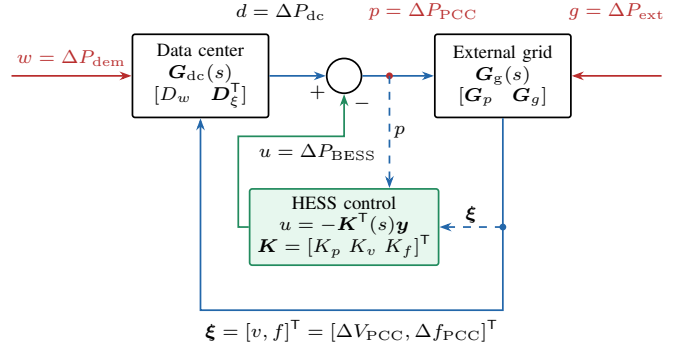

\subsection{Closed-Loop Objectives}

The HESS feedback law in Fig.~\ref{fig:framework} is $u=-\bm{K}^{\mathsf T}(s)\bm{y}=-K_p(s)p-K_v(s)v-K_f(s)f$. Here $\bm{K}(s)=[K_p(s)\ K_v(s)\ K_f(s)]^{\mathsf T}$ includes HESS tracking dynamics for PCC power, voltage, and frequency. Closing the loop in \eqref{eq:port_power}--\eqref{eq:port_measurement} gives the PCC power response
\begin{equation}
  p=\underbrace{T_w(s)w}_{\text{server-demand contribution}}
    +\underbrace{T_g(s)g}_{\text{grid-event contribution}},
  \label{eq:closed_loop_simple}
\end{equation}
where
\begin{equation}
 T_w=P_w-\frac{P_u\bm{K}^{\mathsf T}\bm{h}_w}{1+\bm{K}^{\mathsf T}\bm{h}_u},\quad T_g=P_g-\frac{P_u\bm{K}^{\mathsf T}\bm{h}_g}{1+\bm{K}^{\mathsf T}\bm{h}_u}.
 \label{eq:general_feedback_maps}
\end{equation}
Over the prescribed oscillation band $\Omega$, the HESS should suppress the effect of server-demand changes on grid import while retaining the desired response to grid events.
\begin{equation}
\begin{aligned}
  \text{Isolation:}\quad&T_w=\left.\frac{p}{w}\right|_{g=0}\approx0,\\
  \text{Grid response:}\quad&T_g=\left.\frac{p}{g}\right|_{w=0}\approx T_g^\star,
 \end{aligned}
  \label{eq:two_objectives}
\end{equation}
where the target $T_g^\star$ specifies this desired grid-event response. 

To specify $T_g^\star$ in \eqref{eq:two_objectives}, consider a desired PCC power--frequency characteristic $p=H_f(s)f$ during grid events. With $w=0$, \eqref{eq:closed_loop_simple} gives $p=T_gg$, while the frequency component of \eqref{eq:physical_loop} gives $f=(G_{fp}T_g+G_{fg})g$; $G_{fp}$ and $G_{fg}$ are the frequency entries of $\bm{G}_p$ and $\bm{G}_g$. Thus the desired characteristic requires $T_g^\star=H_f(G_{fp}T_g^\star+G_{fg})$, or
\begin{equation}
 T_g^\star=\frac{H_fG_{fg}}{1-H_fG_{fp}}.
 \label{eq:frequency_support_target}
\end{equation}
For example, positive low-frequency gain of $H_f$ requests reduced grid import when frequency falls.

\section{Fundamental Limits and Trade-Off}
\label{sec:limits}
This section examines what limits the simultaneous objectives in \eqref{eq:two_objectives}. We first establish the isolation--support conflict with PCC-power feedback and identify when additional measurements distinguish the two sources. We then derive the cross-band isolation constraint that remains with multiple measurements and show how sensitivity peaks can amplify grid-response errors under model variation. These results guide measurement selection and performance specifications for the controller design in Section~\ref{sec:controller}.

\subsection{Single-Measurement Isolation-Support Limitation}

The most direct strategy is to measure PCC power and command the HESS to oppose its fluctuations. The difficulty is that a fluctuation in $p$ may originate from either server demand or the grid. When both sources act at the same frequency, a single power measurement cannot distinguish them. To specialize \eqref{eq:general_feedback_maps} to power-only feedback, set $\bm{K}=[K_p,0,0]^{\mathsf T}$, giving $u=-K_pp$. The denominator then becomes $1+\bm{K}^{\mathsf T}\bm{h}_u=1+K_pP_u$. Defining $S=(1+K_pP_u)^{-1}$ and substituting into \eqref{eq:general_feedback_maps} gives
\begin{equation}
 T_w=\frac{P_w}{1+K_pP_u}=SP_w,\quad T_g=\frac{P_g}{1+K_pP_u}=SP_g.
 \label{eq:common_sensitivity}
\end{equation}
Thus the HESS applies the same factor $S$ to both source responses. Making $S$ small isolates server demand, but also suppresses the grid-event power response, even when that response is desired. For a prescribed target $T_g^\star$, define server-demand leakage and grid-response error, normalized by their respective uncontrolled responses, as
\begin{equation}
 A=\left|\frac{T_w}{P_w}\right|=|S|,
 \quad
 E=\left|\frac{T_g-T_g^\star}{P_g}\right|=\left|S-\frac{T_g^\star}{P_g}\right|.
\end{equation}
Here $A$ is the fraction of the uncontrolled PCC-power oscillation that remains under the same server-demand disturbance; $A=0$ denotes complete isolation. The quantity $E$ measures the complex deviation from the prescribed grid response, including both magnitude and phase, relative to $|P_g|$; $E=0$ denotes exact tracking of $T_g^\star$. Both should be small. The following theorem bounds their simultaneous improvement at the same frequency when the two sources are observed through PCC power alone.
\begin{theorem}[Single-port performance limit]
For any prescribed grid-response target $T_g^\star$, a well-posed LTI controller using only PCC power satisfies, at every frequency where the transfers are finite and $P_wP_g\ne0$,
\begin{equation}
A(\omega)+E(\omega)\geq\left|\frac{T_g^\star(j\omega)}{P_g(j\omega)}\right|.
 \label{eq:tradeoff}
\end{equation}
\end{theorem}
\begin{proof}
By \eqref{eq:common_sensitivity} and the triangle inequality, $|T_g^\star/P_g|=|S-(S-T_g^\star/P_g)|\leq|S|+|S-T_g^\star/P_g|=A+E$.
\end{proof}

For response preservation, $T_g^\star=P_g$, the bound reduces to $A+E\geq1$: achieving $A\leq\epsilon$ necessarily gives $E\geq1-\epsilon$. Fig.~\ref{fig:single_port_limit} plots server-demand leakage $A$ (blue solid) and grid-response error $E$ (red dashed) in dB. The low-frequency limits are $A=0.1$ and $E=0.9$, whereas at high frequencies $A\to1$ and $E\to0$. This illustrates the consequence: reducing server-demand-induced PCC ripple to 10\% of its uncontrolled amplitude forces at least 90\% normalized error in the desired grid-event power response at that same frequency. The general bound also covers active-support targets; its right-hand side is not always one.

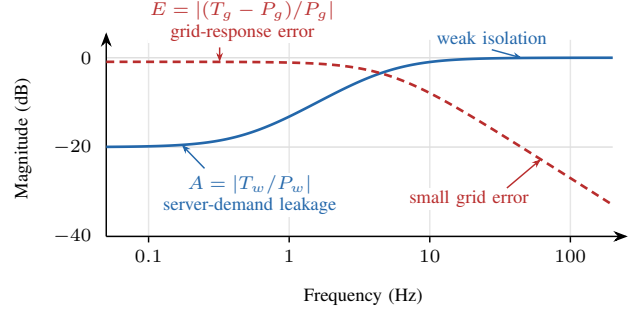
\begin{figure}[!t]
\centering
\begin{tikzpicture}[x=1mm,y=1mm,font=\scriptsize,>=Stealth,line width=0.55pt]
\path[use as bounding box] (0,2) rectangle (85,44);
\foreach \x/\t in {19.60/0.1,38.20/1,56.80/10,75.40/100} {
  \draw[gray!20] (\x,12) -- (\x,38);
  \draw (\x,12) -- (\x,11.2);
  \node[below=0.5mm] at (\x,12) {$\t$};
}
\foreach \y/\t in {12/-40,23.818/-20,35.636/0} {
  \draw[gray!25] (14,\y) -- (81,\y);
  \node[left=0.7mm] at (14,\y) {$\t$};
}
\draw[->] (14,12) -- (83,12);
\draw[->] (14,12) -- (14,39);
\node[rotate=90] at (3,25) {Magnitude (dB)};
\node at (48,4) {Frequency (Hz)};
\begin{scope}
\clip (14,12) rectangle (81,38);
\draw[gridblue,line width=1pt] plot coordinates {(14.000,23.843) (14.419,23.846) (14.838,23.849) (15.256,23.853) (15.675,23.856) (16.094,23.860) (16.512,23.865) (16.931,23.870) (17.350,23.876) (17.769,23.882) (18.188,23.889) (18.606,23.896) (19.025,23.905) (19.444,23.914) (19.863,23.924) (20.281,23.936) (20.700,23.948) (21.119,23.962) (21.538,23.977) (21.956,23.994) (22.375,24.013) (22.794,24.033) (23.212,24.055) (23.631,24.080) (24.050,24.107) (24.469,24.136) (24.887,24.169) (25.306,24.204) (25.725,24.243) (26.144,24.286) (26.563,24.332) (26.981,24.382) (27.400,24.436) (27.819,24.496) (28.238,24.560) (28.656,24.629) (29.075,24.703) (29.494,24.783) (29.913,24.869) (30.331,24.961) (30.750,25.060) (31.169,25.164) (31.587,25.275) (32.006,25.393) (32.425,25.517) (32.844,25.648) (33.263,25.786) (33.681,25.930) (34.100,26.080) (34.519,26.237) (34.938,26.400) (35.356,26.569) (35.775,26.743) (36.194,26.923) (36.612,27.107) (37.031,27.297) (37.450,27.491) (37.869,27.689) (38.288,27.890) (38.706,28.095) (39.125,28.303) (39.544,28.514) (39.962,28.726) (40.381,28.940) (40.800,29.156) (41.219,29.373) (41.638,29.590) (42.056,29.808) (42.475,30.026) (42.894,30.243) (43.313,30.459) (43.731,30.673) (44.150,30.887) (44.569,31.098) (44.987,31.306) (45.406,31.512) (45.825,31.714) (46.244,31.913) (46.663,32.108) (47.081,32.299) (47.500,32.485) (47.919,32.666) (48.337,32.842) (48.756,33.012) (49.175,33.176) (49.594,33.334) (50.013,33.487) (50.431,33.632) (50.850,33.771) (51.269,33.904) (51.688,34.030) (52.106,34.150) (52.525,34.262) (52.944,34.369) (53.362,34.468) (53.781,34.562) (54.200,34.650) (54.619,34.731) (55.038,34.807) (55.456,34.878) (55.875,34.943) (56.294,35.003) (56.712,35.059) (57.131,35.110) (57.550,35.157) (57.969,35.201) (58.388,35.240) (58.806,35.277) (59.225,35.310) (59.644,35.340) (60.063,35.368) (60.481,35.393) (60.900,35.416) (61.319,35.437) (61.737,35.456) (62.156,35.473) (62.575,35.489) (62.994,35.503) (63.413,35.516) (63.831,35.527) (64.250,35.538) (64.669,35.547) (65.088,35.556) (65.506,35.564) (65.925,35.571) (66.344,35.577) (66.763,35.583) (67.181,35.588) (67.600,35.593) (68.019,35.597) (68.438,35.601) (68.856,35.604) (69.275,35.608) (69.694,35.610) (70.112,35.613) (70.531,35.615) (70.950,35.617) (71.369,35.619) (71.787,35.621) (72.206,35.622) (72.625,35.624) (73.044,35.625) (73.463,35.626) (73.881,35.627) (74.300,35.628) (74.719,35.629) (75.138,35.630) (75.556,35.630) (75.975,35.631) (76.394,35.631) (76.813,35.632) (77.231,35.632) (77.650,35.633) (78.069,35.633) (78.487,35.633) (78.906,35.634) (79.325,35.634) (79.744,35.634) (80.162,35.634) (80.581,35.635) (81.000,35.635)};
\draw[tradeoffred,densely dashed,line width=1pt] plot coordinates {(14.000,35.095) (14.419,35.095) (14.838,35.095) (15.256,35.095) (15.675,35.095) (16.094,35.095) (16.512,35.095) (16.931,35.095) (17.350,35.095) (17.769,35.095) (18.188,35.095) (18.606,35.095) (19.025,35.095) (19.444,35.095) (19.863,35.094) (20.281,35.094) (20.700,35.094) (21.119,35.094) (21.538,35.094) (21.956,35.094) (22.375,35.094) (22.794,35.093) (23.212,35.093) (23.631,35.093) (24.050,35.093) (24.469,35.092) (24.887,35.092) (25.306,35.091) (25.725,35.091) (26.144,35.090) (26.563,35.090) (26.981,35.089) (27.400,35.089) (27.819,35.088) (28.238,35.087) (28.656,35.086) (29.075,35.085) (29.494,35.084) (29.913,35.082) (30.331,35.081) (30.750,35.079) (31.169,35.078) (31.587,35.076) (32.006,35.074) (32.425,35.071) (32.844,35.068) (33.263,35.066) (33.681,35.062) (34.100,35.059) (34.519,35.055) (34.938,35.050) (35.356,35.045) (35.775,35.040) (36.194,35.034) (36.612,35.027) (37.031,35.020) (37.450,35.012) (37.869,35.003) (38.288,34.993) (38.706,34.982) (39.125,34.970) (39.544,34.956) (39.962,34.941) (40.381,34.925) (40.800,34.907) (41.219,34.888) (41.638,34.866) (42.056,34.842) (42.475,34.816) (42.894,34.787) (43.313,34.755) (43.731,34.721) (44.150,34.683) (44.569,34.641) (44.987,34.596) (45.406,34.547) (45.825,34.494) (46.244,34.436) (46.663,34.373) (47.081,34.305) (47.500,34.232) (47.919,34.153) (48.337,34.068) (48.756,33.978) (49.175,33.881) (49.594,33.777) (50.013,33.667) (50.431,33.550) (50.850,33.426) (51.269,33.295) (51.688,33.158) (52.106,33.013) (52.525,32.862) (52.944,32.704) (53.362,32.539) (53.781,32.368) (54.200,32.191) (54.619,32.008) (55.038,31.819) (55.456,31.624) (55.875,31.424) (56.294,31.219) (56.712,31.010) (57.131,30.796) (57.550,30.577) (57.969,30.355) (58.388,30.129) (58.806,29.900) (59.225,29.667) (59.644,29.432) (60.063,29.194) (60.481,28.953) (60.900,28.711) (61.319,28.466) (61.737,28.219) (62.156,27.970) (62.575,27.720) (62.994,27.468) (63.413,27.215) (63.831,26.961) (64.250,26.705) (64.669,26.449) (65.088,26.191) (65.506,25.933) (65.925,25.674) (66.344,25.415) (66.763,25.154) (67.181,24.894) (67.600,24.632) (68.019,24.370) (68.438,24.108) (68.856,23.846) (69.275,23.583) (69.694,23.320) (70.112,23.056) (70.531,22.792) (70.950,22.528) (71.369,22.264) (71.787,22.000) (72.206,21.735) (72.625,21.471) (73.044,21.206) (73.463,20.941) (73.881,20.676) (74.300,20.411) (74.719,20.145) (75.138,19.880) (75.556,19.615) (75.975,19.349) (76.394,19.084) (76.813,18.818) (77.231,18.553) (77.650,18.287) (78.069,18.021) (78.487,17.756) (78.906,17.490) (79.325,17.224) (79.744,16.958) (80.162,16.692) (80.581,16.426) (81.000,16.161)};
\end{scope}
\node[text=tradeoffred,align=center,fill=white,inner sep=0.3mm] at (32,40.5) {$E=|(T_g-P_g)/P_g|$\\grid-response error};
\draw[tradeoffred,-{Stealth[length=1.2mm]}] (29,37) -- (29,35.1);
\node[text=gridblue,align=center,fill=white,inner sep=0.3mm] at (33,17.5) {$A=|T_w/P_w|$\\server-demand leakage};
\draw[gridblue,-{Stealth[length=1.2mm]}] (27,21) -- (24,24.1);
\node[text=gridblue,fill=white,inner sep=0.3mm] at (65,38.3) {weak isolation};
\draw[gridblue,-{Stealth[length=1.2mm]}] (65,36.7) -- (69,35.6);
\node[text=tradeoffred,fill=white,inner sep=0.3mm] at (62,17) {small grid error};
\draw[tradeoffred,-{Stealth[length=1.2mm]}] (67,18.4) -- (71.5,22.2);
\end{tikzpicture}
\caption{Single-port limitation with $T_g^\star=P_g$. Blue solid: server-demand leakage $A=|S|$; red dashed: grid-response error $E=|1-S|$. Curves use $S=(s+a)/(s+10a)$, $a=2\pi(0.5)$ rad/s. The bound $A+E\ge1$ applies to linear magnitudes.}
\label{fig:single_port_limit}
\end{figure}

Measuring data-center power $d$ does not directly reveal server demand $w$: voltage and frequency also affect the realized electrical load~\cite[Secs.~II-E and~III-D]{zhao2026datacenter}. Under the ideal balance $d=p+u$, subtracting the HESS contribution from either power measurement gives
\begin{equation}
 p-P_uu=d-(P_u+1)u=P_ww+P_gg.
 \label{eq:power_residual_redundancy}
\end{equation}
Thus the two measurements contain the same mixture of server and grid disturbances, not two independent source signals. For well-posed feedback $u=-K_pp-K_dd$, substituting $d=p+u$ into \eqref{eq:port_power} gives
\begin{equation}
 T_i=\frac{1+K_d}{1+K_d+(K_p+K_d)P_u}P_i,\quad i\in\{w,g\}.
 \label{eq:two_power_maps}
\end{equation}
Both source responses are multiplied by the same factor, so the isolation--support bound \eqref{eq:tradeoff} still applies: adding $d$ changes the feedback dynamics but does not resolve the disturbance origin. An additional measurement must instead provide a distinct source response. After exact actuation subtraction, \eqref{eq:port_measurement} gives $\bm{y}-\bm{h}_uu=\bm{H}[w,g]^{\mathsf T}$, where $\bm{H}=[\bm{h}_w\ \bm{h}_g]$. The following theorem gives the condition where the feedback measurements are enough:

\begin{theorem}[Pointwise source identifiability]
At a fixed frequency and operating point, server-demand and grid disturbances can be algebraically reconstructed from the exact residual if and only if
\begin{equation}
 \operatorname{rank}\bm{H}(j\omega)=2.
 \label{eq:rank_condition}
\end{equation}
\end{theorem}
\begin{proof}
The residual maps $[w,g]^{\mathsf T}$ to $\bm y-\bm h_uu$ injectively exactly when $\bm H$ has full column rank. For $\bm y=[p,v,f]^{\mathsf T}$, subtracting $(P_g/P_w)\bm h_w$ from $\bm h_g$ when $P_w\ne0$ leaves $[0,\bm G_g^{\mathsf T}]^{\mathsf T}$. Hence these measurements satisfy \eqref{eq:rank_condition} precisely when $P_w\ne0$ and $\bm G_g\ne\bm0$.
\end{proof}

\subsection{Cross-Band Trade-Off With Multiple Measurements}

Adding voltage or frequency measurements can distinguish server-demand and grid events when \eqref{eq:rank_condition} holds. However, distinguishing the sources does not allow their responses to be shaped arbitrarily over all frequencies. The measurements still drive the same HESS, whose power changes the PCC conditions and feeds back into those measurements. We now ask whether strong server-demand isolation in one band can be achieved without amplifying server-demand-induced power fluctuations elsewhere. The link to our objective follows directly from \eqref{eq:general_feedback_maps}. For the PCC measurements in \eqref{eq:physical_signatures}, server demand and HESS power affect the measurements through the same net-withdrawal channel, giving $\bm{h}_w=(P_w/P_u)\bm{h}_u$ where $P_u\ne0$. 

Define the HESS return-loop gain and sensitivity as
\begin{equation}
 L_u=\bm{K}^{\mathsf T}\bm{h}_u,\qquad S_u=\frac{1}{1+L_u}.
 \label{eq:return_sensitivity}
\end{equation}
Physically, $L_u$ follows a change in HESS power through the measured PCC quantities and back through the controller. Writing $\bm{K}_\xi=[K_v\ K_f]^{\mathsf T}$, substitution of the signatures into \eqref{eq:general_feedback_maps} gives
\begin{equation}
 \begin{aligned}
  L_u&=P_u(K_p+\bm{K}_\xi^{\mathsf T}\bm{G}_p),\\
  T_w&=P_wS_u,\qquad T_g=(P_g-P_u\bm{K}_\xi^{\mathsf T}\bm{G}_g)S_u.
 \end{aligned}
 \label{eq:two_sensor_maps}
\end{equation}
Thus $|S_u|=|T_w/P_w|$ is the controlled-to-uncontrolled PCC-power amplitude ratio for the same server-demand input, wherever $P_w\ne0$. Values below one indicate attenuation, and values above one indicate amplification. The following theorem quantifies the resulting cross-frequency constraint on server-demand isolation.

\begin{theorem}[Cross-band server-demand isolation limit]
For the PCC model in \eqref{eq:physical_loop}, suppose the interconnection is internally stable, $L_u$ is real-rational and strictly proper, and no unstable modes are hidden by cancellation. Exclude imaginary-axis poles of $L_u$. The sensitivity governing normalized server-demand leakage satisfies
\begin{equation}
 \int_0^\infty\ln|S_u(j\omega)|\,d\omega
 =\pi\sum_i\Re\lambda_i-\frac{\pi}{2}c_\infty,
 \label{eq:bode_general}
\end{equation}
where $\lambda_i$ are the right-half-plane poles of $L_u$, including any controller poles, and $c_\infty=\lim_{s\to\infty}sL_u(s)$.
\end{theorem}

\begin{proof}
Equation~\eqref{eq:two_sensor_maps} identifies normalized server-demand leakage with $S_u=(1+L_u)^{-1}$. Applying the scalar Bode sensitivity integral gives \eqref{eq:bode_general}~\cite{seron1997limits,doyle1992feedback}.
\end{proof}

For a loop with no unstable poles and relative degree at least two, \eqref{eq:bode_general} is zero: reducing server-induced PCC oscillations in one band requires amplification elsewhere. Deeper or wider attenuation, including suppression of several harmonics, requires more positive log-gain area outside the attenuated bands, but does not fix the location or height of individual peaks. Fig.~\ref{fig:bode_area} illustrates this waterbed effect. The PCC-power amplitude is reduced to 11.2\% of its uncontrolled value at 3 Hz, but amplified to 2.195 times that value near 47.71 Hz. Blue and red areas represent attenuation and amplification in \eqref{eq:bode_general}; their full areas balance, including the positive tail beyond 100 Hz. Thus an oscillation moving into the high-frequency peak can be worsened by the same controller that suppresses it at low frequency.

\begin{figure}[!t]
\centering
\begin{tikzpicture}[x=1mm,y=1mm,font=\scriptsize,>=Stealth,line width=0.55pt]
\path[use as bounding box] (0,2) rectangle (85,48);
\fill[gridblue!18] (14.000,29.333) -- plot coordinates {(14.000,29.333) (14.000,29.333) (14.000,29.333) (14.000,29.333) (14.000,29.333) (14.000,29.333) (14.000,29.333) (14.000,29.333) (14.000,29.333) (14.000,29.333) (14.000,29.333) (14.000,29.333) (14.000,29.333) (14.000,29.333) (14.000,29.333) (14.000,29.333) (14.000,29.333) (14.000,29.333) (14.000,29.332) (14.000,29.332) (14.000,29.332) (14.000,29.332) (14.000,29.332) (14.000,29.331) (14.000,29.331) (14.000,29.331) (14.000,29.330) (14.000,29.330) (14.001,29.329) (14.001,29.328) (14.001,29.328) (14.001,29.327) (14.001,29.326) (14.001,29.324) (14.001,29.323) (14.001,29.321) (14.001,29.319) (14.001,29.317) (14.001,29.314) (14.001,29.311) (14.001,29.307) (14.001,29.302) (14.002,29.297) (14.002,29.291) (14.002,29.284) (14.002,29.276) (14.002,29.267) (14.002,29.256) (14.002,29.243) (14.003,29.229) (14.003,29.212) (14.003,29.192) (14.003,29.169) (14.004,29.143) (14.004,29.112) (14.004,29.077) (14.005,29.036) (14.005,28.989) (14.005,28.935) (14.006,28.874) (14.006,28.803) (14.007,28.723) (14.007,28.632) (14.008,28.529) (14.008,28.413) (14.009,28.283) (14.010,28.138) (14.011,27.976) (14.011,27.797) (14.012,27.600) (14.013,27.385) (14.014,27.152) (14.016,26.899) (14.017,26.628) (14.018,26.339) (14.020,26.033) (14.021,25.710) (14.023,25.372) (14.025,25.019) (14.027,24.653) (14.029,24.277) (14.031,23.890) (14.034,23.496) (14.036,23.095) (14.039,22.689) (14.042,22.281) (14.046,21.872) (14.049,21.463) (14.053,21.056) (14.058,20.653) (14.062,20.256) (14.067,19.867) (14.073,19.486) (14.078,19.117) (14.085,18.759) (14.091,18.416) (14.099,18.087) (14.106,17.775) (14.115,17.479) (14.124,17.202) (14.134,16.943) (14.134,16.943) (14.236,15.590) (14.338,15.153) (14.441,14.965) (14.543,14.869) (14.645,14.814) (14.747,14.780) (14.849,14.758) (14.952,14.743) (15.054,14.732) (15.156,14.724) (15.258,14.719) (15.360,14.715) (15.463,14.713) (15.565,14.712) (15.667,14.711) (15.769,14.711) (15.871,14.712) (15.973,14.713) (16.076,14.715) (16.178,14.717) (16.280,14.720) (16.382,14.723) (16.484,14.727) (16.587,14.731) (16.689,14.736) (16.791,14.741) (16.893,14.746) (16.995,14.753) (17.098,14.759) (17.200,14.766) (17.302,14.774) (17.404,14.782) (17.506,14.791) (17.609,14.801) (17.711,14.811) (17.813,14.822) (17.915,14.833) (18.017,14.845) (18.120,14.858) (18.222,14.872) (18.324,14.887) (18.426,14.902) (18.528,14.919) (18.631,14.936) (18.733,14.954) (18.835,14.973) (18.937,14.994) (19.039,15.015) (19.141,15.037) (19.244,15.060) (19.346,15.085) (19.448,15.110) (19.550,15.137) (19.652,15.165) (19.755,15.194) (19.857,15.224) (19.959,15.255) (20.061,15.288) (20.163,15.322) (20.266,15.357) (20.368,15.394) (20.470,15.432) (20.572,15.471) (20.674,15.512) (20.777,15.554) (20.879,15.597) (20.981,15.642) (21.083,15.688) (21.185,15.736) (21.288,15.785) (21.390,15.835) (21.492,15.887) (21.594,15.940) (21.696,15.994) (21.799,16.050) (21.901,16.108) (22.003,16.166) (22.105,16.226) (22.207,16.287) (22.309,16.350) (22.412,16.414) (22.514,16.479) (22.616,16.546) (22.718,16.613) (22.820,16.682) (22.923,16.752) (23.025,16.824) (23.127,16.896) (23.229,16.970) (23.331,17.044) (23.434,17.120) (23.536,17.197) (23.638,17.275) (23.740,17.354) (23.842,17.433) (23.945,17.514) (24.047,17.596) (24.149,17.678) (24.251,17.762) (24.353,17.846) (24.456,17.931) (24.558,18.017) (24.660,18.104) (24.762,18.191) (24.864,18.279) (24.967,18.367) (25.069,18.457) (25.171,18.547) (25.273,18.637) (25.375,18.728) (25.477,18.820) (25.580,18.912) (25.682,19.004) (25.784,19.097) (25.886,19.191) (25.988,19.284) (26.091,19.379) (26.193,19.473) (26.295,19.568) (26.397,19.664) (26.499,19.759) (26.602,19.855) (26.704,19.952) (26.806,20.048) (26.908,20.145) (27.010,20.242) (27.113,20.339) (27.215,20.437) (27.317,20.534) (27.419,20.632) (27.521,20.730) (27.624,20.828) (27.726,20.927) (27.828,21.025) (27.930,21.123) (28.032,21.222) (28.135,21.321) (28.237,21.420) (28.339,21.519) (28.441,21.618) (28.543,21.717) (28.645,21.816) (28.748,21.915) (28.850,22.014) (28.952,22.113) (29.054,22.213) (29.156,22.312) (29.259,22.411) (29.361,22.510) (29.463,22.610) (29.565,22.709) (29.667,22.808) (29.770,22.908) (29.872,23.007) (29.974,23.106) (30.076,23.206) (30.178,23.305) (30.281,23.404) (30.383,23.503) (30.485,23.602) (30.587,23.702) (30.689,23.801) (30.792,23.900) (30.894,23.999) (30.996,24.098) (31.098,24.197) (31.200,24.296) (31.303,24.394) (31.405,24.493) (31.507,24.592) (31.609,24.691) (31.711,24.789) (31.813,24.888) (31.916,24.986) (32.018,25.085) (32.120,25.183) (32.222,25.282) (32.324,25.380) (32.427,25.478) (32.529,25.576) (32.631,25.675) (32.733,25.773) (32.835,25.871) (32.938,25.968) (33.040,26.066) (33.142,26.164) (33.244,26.262) (33.346,26.359) (33.449,26.457) (33.551,26.554) (33.653,26.652) (33.755,26.749) (33.857,26.846) (33.960,26.944) (34.062,27.041) (34.164,27.138) (34.266,27.234) (34.368,27.331) (34.471,27.428) (34.573,27.524) (34.675,27.621) (34.777,27.717) (34.879,27.814) (34.981,27.910) (35.084,28.006) (35.186,28.102) (35.288,28.197) (35.390,28.293) (35.492,28.388) (35.595,28.484) (35.697,28.579) (35.799,28.674) (35.901,28.769) (36.003,28.863) (36.106,28.958) (36.208,29.052) (36.310,29.146) (36.412,29.240) (36.514,29.333)} -- (36.514,29.333) -- cycle;
\fill[tradeoffred!18] (36.514,29.333) -- plot coordinates {(36.514,29.333) (36.701,29.503) (36.887,29.672) (37.073,29.841) (37.259,30.008) (37.445,30.174) (37.631,30.339) (37.817,30.503) (38.003,30.665) (38.190,30.826) (38.376,30.986) (38.562,31.144) (38.748,31.300) (38.934,31.454) (39.120,31.606) (39.306,31.756) (39.493,31.903) (39.679,32.048) (39.865,32.190) (40.051,32.330) (40.237,32.467) (40.423,32.600) (40.609,32.730) (40.795,32.857) (40.982,32.979) (41.168,33.098) (41.354,33.214) (41.540,33.325) (41.726,33.431) (41.912,33.533) (42.098,33.631) (42.284,33.724) (42.471,33.812) (42.657,33.896) (42.843,33.974) (43.029,34.048) (43.215,34.116) (43.401,34.180) (43.587,34.238) (43.774,34.291) (43.960,34.340) (44.146,34.383) (44.332,34.422) (44.518,34.456) (44.704,34.485) (44.890,34.510) (45.076,34.531) (45.263,34.547) (45.449,34.559) (45.635,34.567) (45.821,34.572) (46.007,34.573) (46.193,34.571) (46.379,34.565) (46.566,34.556) (46.752,34.545) (46.938,34.530) (47.124,34.514) (47.310,34.495) (47.496,34.474) (47.682,34.450) (47.868,34.425) (48.055,34.399) (48.241,34.370) (48.427,34.340) (48.613,34.309) (48.799,34.277) (48.985,34.244) (49.171,34.210) (49.358,34.175) (49.544,34.139) (49.730,34.103) (49.916,34.066) (50.102,34.028) (50.288,33.990) (50.474,33.952) (50.660,33.914) (50.847,33.875) (51.033,33.836) (51.219,33.797) (51.405,33.758) (51.591,33.719) (51.777,33.680) (51.963,33.641) (52.150,33.602) (52.336,33.564) (52.522,33.525) (52.708,33.487) (52.894,33.449) (53.080,33.411) (53.266,33.373) (53.452,33.335) (53.639,33.298) (53.825,33.261) (54.011,33.225) (54.197,33.188) (54.383,33.152) (54.569,33.117) (54.755,33.081) (54.941,33.046) (55.128,33.012) (55.314,32.977) (55.500,32.943) (55.686,32.910) (55.872,32.877) (56.058,32.844) (56.244,32.811) (56.431,32.779) (56.617,32.747) (56.803,32.716) (56.989,32.685) (57.175,32.654) (57.361,32.623) (57.547,32.593) (57.733,32.564) (57.920,32.534) (58.106,32.505) (58.292,32.477) (58.478,32.448) (58.664,32.420) (58.850,32.393) (59.036,32.366) (59.223,32.339) (59.409,32.312) (59.595,32.286) (59.781,32.260) (59.967,32.234) (60.153,32.209) (60.339,32.183) (60.525,32.159) (60.712,32.134) (60.898,32.110) (61.084,32.086) (61.270,32.063) (61.456,32.039) (61.642,32.016) (61.828,31.994) (62.015,31.971) (62.201,31.949) (62.387,31.927) (62.573,31.905) (62.759,31.884) (62.945,31.863) (63.131,31.842) (63.317,31.821) (63.504,31.801) (63.690,31.781) (63.876,31.761) (64.062,31.741) (64.248,31.722) (64.434,31.702) (64.620,31.684) (64.806,31.665) (64.993,31.646) (65.179,31.628) (65.365,31.610) (65.551,31.592) (65.737,31.574) (65.923,31.557) (66.109,31.539) (66.296,31.522) (66.482,31.505) (66.668,31.489) (66.854,31.472) (67.040,31.456) (67.226,31.440) (67.412,31.424) (67.598,31.408) (67.785,31.392) (67.971,31.377) (68.157,31.362) (68.343,31.347) (68.529,31.332) (68.715,31.317) (68.901,31.302) (69.088,31.288) (69.274,31.274) (69.460,31.260) (69.646,31.246) (69.832,31.232) (70.018,31.218) (70.204,31.205) (70.390,31.191) (70.577,31.178) (70.763,31.165) (70.949,31.152) (71.135,31.139) (71.321,31.127) (71.507,31.114) (71.693,31.102) (71.880,31.089) (72.066,31.077) (72.252,31.065) (72.438,31.053) (72.624,31.042) (72.810,31.030) (72.996,31.018) (73.182,31.007) (73.369,30.996) (73.555,30.985) (73.741,30.973) (73.927,30.962) (74.113,30.952) (74.299,30.941) (74.485,30.930) (74.672,30.920) (74.858,30.909) (75.044,30.899) (75.230,30.889) (75.416,30.879) (75.602,30.869) (75.788,30.859) (75.974,30.849) (76.161,30.839) (76.347,30.830) (76.533,30.820) (76.719,30.811) (76.905,30.801) (77.091,30.792) (77.277,30.783) (77.463,30.774) (77.650,30.765) (77.836,30.756) (78.022,30.747) (78.208,30.738) (78.394,30.730) (78.580,30.721) (78.766,30.712) (78.953,30.704) (79.139,30.696) (79.325,30.687) (79.511,30.679) (79.697,30.671) (79.883,30.663) (80.069,30.655) (80.255,30.647) (80.442,30.639) (80.628,30.632) (80.814,30.624) (81.000,30.616)} -- (81.000,29.333) -- cycle;
\foreach \x/\t in {14/0,27.4/20,40.8/40,54.2/60,67.6/80,81/100} {
  \draw[gray!20] (\x,12) -- (\x,38);
  \draw (\x,12) -- (\x,11.2);
  \node[below=0.5mm] at (\x,12) {$\t$};
}
\foreach \y/\t in {16/-2,22.667/-1,29.333/0,36/1} {
  \draw[gray!25] (14,\y) -- (81,\y);
  \node[left=0.7mm] at (14,\y) {$\t$};
}
\draw[supportgreen,densely dotted] (14.67,12) -- (14.67,39);
\draw[supportgreen,densely dotted] (27.4,12) -- (27.4,39);
\draw[supportgreen,{Stealth[length=1.2mm]}-{Stealth[length=1.2mm]}] (14.67,40) -- (27.4,40);
\node[text=supportgreen,align=center,anchor=south] at (26,40.8) {1--20 Hz isolation band};
\draw[gray,dashed] (14,29.333) -- (83,29.333);
\draw[->] (14,12) -- (83,12);
\draw[->] (14,12) -- (14,39);
\draw[gridblue,line width=1pt] plot coordinates {(14.000,29.333) (14.000,29.333) (14.000,29.333) (14.000,29.333) (14.000,29.333) (14.000,29.333) (14.000,29.333) (14.000,29.333) (14.000,29.333) (14.000,29.333) (14.000,29.333) (14.000,29.333) (14.000,29.333) (14.000,29.333) (14.000,29.333) (14.000,29.333) (14.000,29.333) (14.000,29.333) (14.000,29.332) (14.000,29.332) (14.000,29.332) (14.000,29.332) (14.000,29.332) (14.000,29.331) (14.000,29.331) (14.000,29.331) (14.000,29.330) (14.000,29.330) (14.001,29.329) (14.001,29.328) (14.001,29.328) (14.001,29.327) (14.001,29.326) (14.001,29.324) (14.001,29.323) (14.001,29.321) (14.001,29.319) (14.001,29.317) (14.001,29.314) (14.001,29.311) (14.001,29.307) (14.001,29.302) (14.002,29.297) (14.002,29.291) (14.002,29.284) (14.002,29.276) (14.002,29.267) (14.002,29.256) (14.002,29.243) (14.003,29.229) (14.003,29.212) (14.003,29.192) (14.003,29.169) (14.004,29.143) (14.004,29.112) (14.004,29.077) (14.005,29.036) (14.005,28.989) (14.005,28.935) (14.006,28.874) (14.006,28.803) (14.007,28.723) (14.007,28.632) (14.008,28.529) (14.008,28.413) (14.009,28.283) (14.010,28.138) (14.011,27.976) (14.011,27.797) (14.012,27.600) (14.013,27.385) (14.014,27.152) (14.016,26.899) (14.017,26.628) (14.018,26.339) (14.020,26.033) (14.021,25.710) (14.023,25.372) (14.025,25.019) (14.027,24.653) (14.029,24.277) (14.031,23.890) (14.034,23.496) (14.036,23.095) (14.039,22.689) (14.042,22.281) (14.046,21.872) (14.049,21.463) (14.053,21.056) (14.058,20.653) (14.062,20.256) (14.067,19.867) (14.073,19.486) (14.078,19.117) (14.085,18.759) (14.091,18.416) (14.099,18.087) (14.106,17.775) (14.115,17.479) (14.124,17.202) (14.134,16.943) (14.134,16.943) (14.236,15.590) (14.338,15.153) (14.441,14.965) (14.543,14.869) (14.645,14.814) (14.747,14.780) (14.849,14.758) (14.952,14.743) (15.054,14.732) (15.156,14.724) (15.258,14.719) (15.360,14.715) (15.463,14.713) (15.565,14.712) (15.667,14.711) (15.769,14.711) (15.871,14.712) (15.973,14.713) (16.076,14.715) (16.178,14.717) (16.280,14.720) (16.382,14.723) (16.484,14.727) (16.587,14.731) (16.689,14.736) (16.791,14.741) (16.893,14.746) (16.995,14.753) (17.098,14.759) (17.200,14.766) (17.302,14.774) (17.404,14.782) (17.506,14.791) (17.609,14.801) (17.711,14.811) (17.813,14.822) (17.915,14.833) (18.017,14.845) (18.120,14.858) (18.222,14.872) (18.324,14.887) (18.426,14.902) (18.528,14.919) (18.631,14.936) (18.733,14.954) (18.835,14.973) (18.937,14.994) (19.039,15.015) (19.141,15.037) (19.244,15.060) (19.346,15.085) (19.448,15.110) (19.550,15.137) (19.652,15.165) (19.755,15.194) (19.857,15.224) (19.959,15.255) (20.061,15.288) (20.163,15.322) (20.266,15.357) (20.368,15.394) (20.470,15.432) (20.572,15.471) (20.674,15.512) (20.777,15.554) (20.879,15.597) (20.981,15.642) (21.083,15.688) (21.185,15.736) (21.288,15.785) (21.390,15.835) (21.492,15.887) (21.594,15.940) (21.696,15.994) (21.799,16.050) (21.901,16.108) (22.003,16.166) (22.105,16.226) (22.207,16.287) (22.309,16.350) (22.412,16.414) (22.514,16.479) (22.616,16.546) (22.718,16.613) (22.820,16.682) (22.923,16.752) (23.025,16.824) (23.127,16.896) (23.229,16.970) (23.331,17.044) (23.434,17.120) (23.536,17.197) (23.638,17.275) (23.740,17.354) (23.842,17.433) (23.945,17.514) (24.047,17.596) (24.149,17.678) (24.251,17.762) (24.353,17.846) (24.456,17.931) (24.558,18.017) (24.660,18.104) (24.762,18.191) (24.864,18.279) (24.967,18.367) (25.069,18.457) (25.171,18.547) (25.273,18.637) (25.375,18.728) (25.477,18.820) (25.580,18.912) (25.682,19.004) (25.784,19.097) (25.886,19.191) (25.988,19.284) (26.091,19.379) (26.193,19.473) (26.295,19.568) (26.397,19.664) (26.499,19.759) (26.602,19.855) (26.704,19.952) (26.806,20.048) (26.908,20.145) (27.010,20.242) (27.113,20.339) (27.215,20.437) (27.317,20.534) (27.419,20.632) (27.521,20.730) (27.624,20.828) (27.726,20.927) (27.828,21.025) (27.930,21.123) (28.032,21.222) (28.135,21.321) (28.237,21.420) (28.339,21.519) (28.441,21.618) (28.543,21.717) (28.645,21.816) (28.748,21.915) (28.850,22.014) (28.952,22.113) (29.054,22.213) (29.156,22.312) (29.259,22.411) (29.361,22.510) (29.463,22.610) (29.565,22.709) (29.667,22.808) (29.770,22.908) (29.872,23.007) (29.974,23.106) (30.076,23.206) (30.178,23.305) (30.281,23.404) (30.383,23.503) (30.485,23.602) (30.587,23.702) (30.689,23.801) (30.792,23.900) (30.894,23.999) (30.996,24.098) (31.098,24.197) (31.200,24.296) (31.303,24.394) (31.405,24.493) (31.507,24.592) (31.609,24.691) (31.711,24.789) (31.813,24.888) (31.916,24.986) (32.018,25.085) (32.120,25.183) (32.222,25.282) (32.324,25.380) (32.427,25.478) (32.529,25.576) (32.631,25.675) (32.733,25.773) (32.835,25.871) (32.938,25.968) (33.040,26.066) (33.142,26.164) (33.244,26.262) (33.346,26.359) (33.449,26.457) (33.551,26.554) (33.653,26.652) (33.755,26.749) (33.857,26.846) (33.960,26.944) (34.062,27.041) (34.164,27.138) (34.266,27.234) (34.368,27.331) (34.471,27.428) (34.573,27.524) (34.675,27.621) (34.777,27.717) (34.879,27.814) (34.981,27.910) (35.084,28.006) (35.186,28.102) (35.288,28.197) (35.390,28.293) (35.492,28.388) (35.595,28.484) (35.697,28.579) (35.799,28.674) (35.901,28.769) (36.003,28.863) (36.106,28.958) (36.208,29.052) (36.310,29.146) (36.412,29.240) (36.514,29.333)};
\draw[tradeoffred,line width=1pt] plot coordinates {(36.514,29.333) (36.701,29.503) (36.887,29.672) (37.073,29.841) (37.259,30.008) (37.445,30.174) (37.631,30.339) (37.817,30.503) (38.003,30.665) (38.190,30.826) (38.376,30.986) (38.562,31.144) (38.748,31.300) (38.934,31.454) (39.120,31.606) (39.306,31.756) (39.493,31.903) (39.679,32.048) (39.865,32.190) (40.051,32.330) (40.237,32.467) (40.423,32.600) (40.609,32.730) (40.795,32.857) (40.982,32.979) (41.168,33.098) (41.354,33.214) (41.540,33.325) (41.726,33.431) (41.912,33.533) (42.098,33.631) (42.284,33.724) (42.471,33.812) (42.657,33.896) (42.843,33.974) (43.029,34.048) (43.215,34.116) (43.401,34.180) (43.587,34.238) (43.774,34.291) (43.960,34.340) (44.146,34.383) (44.332,34.422) (44.518,34.456) (44.704,34.485) (44.890,34.510) (45.076,34.531) (45.263,34.547) (45.449,34.559) (45.635,34.567) (45.821,34.572) (46.007,34.573) (46.193,34.571) (46.379,34.565) (46.566,34.556) (46.752,34.545) (46.938,34.530) (47.124,34.514) (47.310,34.495) (47.496,34.474) (47.682,34.450) (47.868,34.425) (48.055,34.399) (48.241,34.370) (48.427,34.340) (48.613,34.309) (48.799,34.277) (48.985,34.244) (49.171,34.210) (49.358,34.175) (49.544,34.139) (49.730,34.103) (49.916,34.066) (50.102,34.028) (50.288,33.990) (50.474,33.952) (50.660,33.914) (50.847,33.875) (51.033,33.836) (51.219,33.797) (51.405,33.758) (51.591,33.719) (51.777,33.680) (51.963,33.641) (52.150,33.602) (52.336,33.564) (52.522,33.525) (52.708,33.487) (52.894,33.449) (53.080,33.411) (53.266,33.373) (53.452,33.335) (53.639,33.298) (53.825,33.261) (54.011,33.225) (54.197,33.188) (54.383,33.152) (54.569,33.117) (54.755,33.081) (54.941,33.046) (55.128,33.012) (55.314,32.977) (55.500,32.943) (55.686,32.910) (55.872,32.877) (56.058,32.844) (56.244,32.811) (56.431,32.779) (56.617,32.747) (56.803,32.716) (56.989,32.685) (57.175,32.654) (57.361,32.623) (57.547,32.593) (57.733,32.564) (57.920,32.534) (58.106,32.505) (58.292,32.477) (58.478,32.448) (58.664,32.420) (58.850,32.393) (59.036,32.366) (59.223,32.339) (59.409,32.312) (59.595,32.286) (59.781,32.260) (59.967,32.234) (60.153,32.209) (60.339,32.183) (60.525,32.159) (60.712,32.134) (60.898,32.110) (61.084,32.086) (61.270,32.063) (61.456,32.039) (61.642,32.016) (61.828,31.994) (62.015,31.971) (62.201,31.949) (62.387,31.927) (62.573,31.905) (62.759,31.884) (62.945,31.863) (63.131,31.842) (63.317,31.821) (63.504,31.801) (63.690,31.781) (63.876,31.761) (64.062,31.741) (64.248,31.722) (64.434,31.702) (64.620,31.684) (64.806,31.665) (64.993,31.646) (65.179,31.628) (65.365,31.610) (65.551,31.592) (65.737,31.574) (65.923,31.557) (66.109,31.539) (66.296,31.522) (66.482,31.505) (66.668,31.489) (66.854,31.472) (67.040,31.456) (67.226,31.440) (67.412,31.424) (67.598,31.408) (67.785,31.392) (67.971,31.377) (68.157,31.362) (68.343,31.347) (68.529,31.332) (68.715,31.317) (68.901,31.302) (69.088,31.288) (69.274,31.274) (69.460,31.260) (69.646,31.246) (69.832,31.232) (70.018,31.218) (70.204,31.205) (70.390,31.191) (70.577,31.178) (70.763,31.165) (70.949,31.152) (71.135,31.139) (71.321,31.127) (71.507,31.114) (71.693,31.102) (71.880,31.089) (72.066,31.077) (72.252,31.065) (72.438,31.053) (72.624,31.042) (72.810,31.030) (72.996,31.018) (73.182,31.007) (73.369,30.996) (73.555,30.985) (73.741,30.973) (73.927,30.962) (74.113,30.952) (74.299,30.941) (74.485,30.930) (74.672,30.920) (74.858,30.909) (75.044,30.899) (75.230,30.889) (75.416,30.879) (75.602,30.869) (75.788,30.859) (75.974,30.849) (76.161,30.839) (76.347,30.830) (76.533,30.820) (76.719,30.811) (76.905,30.801) (77.091,30.792) (77.277,30.783) (77.463,30.774) (77.650,30.765) (77.836,30.756) (78.022,30.747) (78.208,30.738) (78.394,30.730) (78.580,30.721) (78.766,30.712) (78.953,30.704) (79.139,30.696) (79.325,30.687) (79.511,30.679) (79.697,30.671) (79.883,30.663) (80.069,30.655) (80.255,30.647) (80.442,30.639) (80.628,30.632) (80.814,30.624) (81.000,30.616)};
\fill[tradeoffred] (45.966,34.573) circle (0.45mm);
\node[text=tradeoffred,align=center,fill=white,inner sep=0.3mm] at (64,35.8) {$|S_u|=2.195$\\47.71 Hz};
\draw[tradeoffred,-{Stealth[length=1.2mm]}] (54.5,35.5) -- (46.766,34.573);
\node[text=gridblue,align=center,fill=white,inner sep=0.3mm] at (40,22.5) {attenuation\\$\ln|S_u|<0$};
\draw[gridblue,-{Stealth[length=1.2mm]}] (31.5,23.5) -- (24,25);
\draw[tradeoffred,densely dotted,->] (81.000,30.616) -- (84.000,30.514);
\node[text=tradeoffred,align=center,anchor=east,fill=white,inner sep=0.3mm] at (79,24.4) {amplification\\$\ln|S_u|>0$};
\draw[black!60,-{Stealth[length=1.2mm]}] (79,26) -- (82.700,30.351);
\node[rotate=90] at (3,25) {$\ln|S_u(j\omega)|$};
\node at (47,4) {Frequency (Hz)};
\end{tikzpicture}
\caption{Waterbed effect with $S_u=1/(1+F)$, $F(s)=\frac{8s}{s+\omega_h}\frac{\omega_b^2}{s^2+1.4\omega_bs+\omega_b^2}$, $\omega_h=2\pi(0.2)$, and $\omega_b=2\pi(15)$ rad/s. Blue: attenuation; red: amplification. Their full log-gain areas balance, including the tail beyond 100 Hz. The linear Hz axis scales both areas by the same factor relative to \eqref{eq:bode_general}.}
\label{fig:bode_area}
\end{figure}
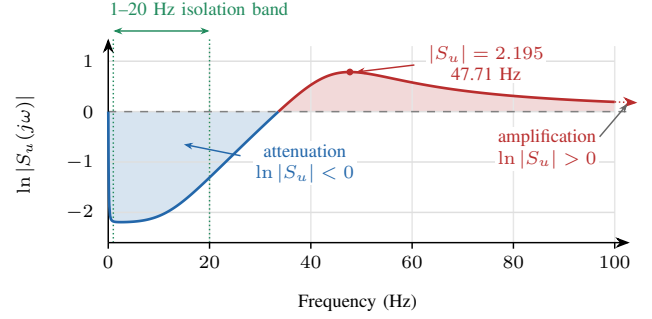

\subsection{Cross-Channel Coupling and Grid-Response Shaping}
\label{sec:cross_channel}

The preceding subsection quantified the cross-frequency limitation within the server-demand channel $T_w$. Its attenuation and out-of-band amplification can also affect the grid-event channel $T_g$, because the same $S_u$ multiplies both responses in \eqref{eq:two_sensor_maps}. Tracking a prescribed target $T_g^\star$ therefore requires the grid-channel numerator to compensate the shaping imposed by $S_u$, including phase. We first describe this nominal matching and then examine how model or operating-point changes disturb it.

Throughout the following analysis, superscript $0$ denotes nominal-model quantities; transfers without this superscript describe the actual plant and its closed loop. Suppose a fixed controller $\bm{K}$ achieves the prescribed target on the nominal model over the frequencies of interest. By \eqref{eq:two_sensor_maps}, exact nominal tracking requires
\begin{equation}
 P_g^0-P_u^0\bm{K}_\xi^{\mathsf T}\bm{G}_g^0
 =(1+L_u^0)T_g^\star.
 \label{eq:nominal_grid_match}
\end{equation}
Thus, exact matching can prevent the isolation sensitivity from distorting the desired grid response on the nominal model.

Let $\delta X=X-X^0$ denote the difference between an actual transfer and its nominal counterpart. Keeping the controller and target fixed, substitution of \eqref{eq:nominal_grid_match} into \eqref{eq:two_sensor_maps} gives
\begin{equation}
 \begin{aligned}
 T_g-T_g^\star=S_u\big[\delta P_g-T_g^\star\delta L_u -\delta(P_u\bm{K}_\xi^{\mathsf T}\bm{G}_g)\big].
 \end{aligned}
 \label{eq:grid_tracking_mismatch}
\end{equation}

The bracketed term is the matching residual created by model or operating-point variation, and $S_u$ is evaluated for the actual closed loop. A sensitivity peak outside the isolation band magnifies a nonzero residual at the same frequency. Thus, nominal grid-response matching can become sensitive to model error in the frequency range where isolation incurs its waterbed cost.

Fig.~\ref{fig:cross_band_limit} illustrates this mismatch-amplification mechanism using the scalar model specified in its caption, with $T_g^\star=P_g$ as the illustrative target. Panel (a) shows similar isolation and out-of-band peaks with exact and mismatched load compensation, whereas panel (b) shows grid-response distortion when the nominal compensation parameter is $D_v^0=1.3D_v$. This is the sensitivity-weighted matching error described by \eqref{eq:grid_tracking_mismatch}. The impact on grid voltage and frequency follows from $\bm\xi=(\bm G_pT_g+\bm G_g)g$ in \eqref{eq:physical_loop}, using both magnitude and phase. When server-demand isolation is required at low frequencies and grid-response accuracy at higher frequencies, these requirements can conflict. If an out-of-band sensitivity peak lies in the required grid-response band, it amplifies the matching residual precisely where grid-support accuracy is needed. Controller design should therefore account for both frequency bands and model variations across operating points.

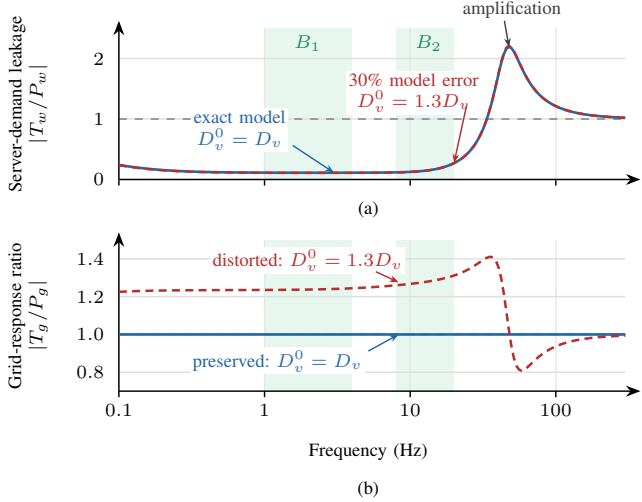
\begin{figure}[!t]
\centering
\begin{tikzpicture}[x=1mm,y=1mm,font=\scriptsize,>=Stealth,line width=0.55pt]
\path[use as bounding box] (0,-2) rectangle (85,65);
\fill[lightgreen] (33.269,40) rectangle (44.870,60);
\fill[lightgreen] (50.670,40) rectangle (58.339,60);
\fill[lightgreen] (33.269,12) rectangle (44.870,32);
\fill[lightgreen] (50.670,12) rectangle (58.339,32);
\foreach \x/\t in {14.000/0.1,33.269/1,52.537/10,71.806/100} {
  \draw[gray!20] (\x,40) -- (\x,60);
  \draw[gray!20] (\x,12) -- (\x,32);
  \draw (\x,40) -- (\x,39.2);
  \draw (\x,12) -- (\x,11.2);
  \node[below=0.5mm] at (\x,12) {$\t$};
}
\foreach \y/\t in {40/0,48/1,56/2} {
  \draw[gray!25] (14,\y) -- (81,\y);
  \node[left=0.7mm] at (14,\y) {$\t$};
}
\foreach \y/\t in {14.5/0.8,19.5/1.0,24.5/1.2,29.5/1.4} {
  \draw[gray!25] (14,\y) -- (81,\y);
  \node[left=0.7mm] at (14,\y) {$\t$};
}
\draw[gray,dashed] (14,48) -- (81,48);
\draw[gray,dashed] (14,19.5) -- (81,19.5);
\draw[->] (14,40) -- (83,40);
\draw[->] (14,40) -- (14,61);
\draw[->] (14,12) -- (83,12);
\draw[->] (14,12) -- (14,33);
\node[rotate=90,align=center] at (2,50) {Server-demand leakage\\$|T_w/P_w|$};
\node[rotate=90,align=center] at (2,22) {Grid-response ratio\\$|T_g/P_g|$};
\node[text=supportgreen] at (39,58) {$B_1$};
\node[text=supportgreen] at (55,58) {$B_2$};
\begin{scope}
\clip (14,40) rectangle (81,60);
\draw[gridblue,line width=1pt] plot coordinates {(14.000,41.937) (14.335,41.880) (14.670,41.825) (15.005,41.773) (15.340,41.723) (15.675,41.675) (16.010,41.629) (16.345,41.586) (16.680,41.544) (17.015,41.505) (17.350,41.468) (17.685,41.432) (18.020,41.398) (18.355,41.366) (18.690,41.336) (19.025,41.307) (19.360,41.280) (19.695,41.255) (20.030,41.231) (20.365,41.208) (20.700,41.187) (21.035,41.167) (21.370,41.148) (21.705,41.130) (22.040,41.113) (22.375,41.098) (22.710,41.083) (23.045,41.069) (23.380,41.056) (23.715,41.044) (24.050,41.033) (24.385,41.023) (24.720,41.013) (25.055,41.004) (25.390,40.996) (25.725,40.988) (26.060,40.980) (26.395,40.974) (26.730,40.967) (27.065,40.961) (27.400,40.956) (27.735,40.951) (28.070,40.946) (28.405,40.942) (28.740,40.938) (29.075,40.934) (29.410,40.930) (29.745,40.927) (30.080,40.924) (30.415,40.921) (30.750,40.919) (31.085,40.916) (31.420,40.914) (31.755,40.912) (32.090,40.910) (32.425,40.909) (32.760,40.907) (33.095,40.906) (33.430,40.904) (33.765,40.903) (34.100,40.902) (34.435,40.901) (34.770,40.900) (35.105,40.899) (35.440,40.898) (35.775,40.897) (36.110,40.897) (36.445,40.896) (36.780,40.896) (37.115,40.895) (37.450,40.895) (37.785,40.894) (38.120,40.894) (38.455,40.893) (38.790,40.893) (39.125,40.893) (39.460,40.893) (39.795,40.893) (40.130,40.892) (40.465,40.892) (40.800,40.892) (41.135,40.892) (41.470,40.892) (41.805,40.892) (42.140,40.893) (42.475,40.893) (42.810,40.893) (43.145,40.893) (43.480,40.894) (43.815,40.894) (44.150,40.894) (44.485,40.895) (44.820,40.896) (45.155,40.896) (45.490,40.897) (45.825,40.898) (46.160,40.899) (46.495,40.900) (46.830,40.902) (47.165,40.903) (47.500,40.905) (47.835,40.907) (48.170,40.910) (48.505,40.913) (48.840,40.916) (49.175,40.920) (49.510,40.924) (49.845,40.929) (50.180,40.934) (50.515,40.941) (50.850,40.948) (51.185,40.957) (51.520,40.967) (51.855,40.979) (52.190,40.992) (52.525,41.007) (52.860,41.025) (53.195,41.045) (53.530,41.069) (53.865,41.096) (54.200,41.127) (54.535,41.162) (54.870,41.203) (55.205,41.250) (55.540,41.304) (55.875,41.365) (56.210,41.435) (56.545,41.514) (56.880,41.605) (57.215,41.707) (57.550,41.824) (57.885,41.956) (58.220,42.106) (58.555,42.275) (58.890,42.468) (59.225,42.687) (59.560,42.937) (59.895,43.221) (60.230,43.545) (60.565,43.916) (60.900,44.342) (61.235,44.833) (61.570,45.399) (61.905,46.055) (62.240,46.815) (62.575,47.695) (62.910,48.712) (63.245,49.874) (63.580,51.180) (63.915,52.599) (64.250,54.061) (64.585,55.439) (64.920,56.567) (65.255,57.296) (65.590,57.555) (65.925,57.387) (66.260,56.910) (66.595,56.256) (66.930,55.533) (67.265,54.811) (67.600,54.127) (67.935,53.500) (68.270,52.934) (68.605,52.429) (68.940,51.980) (69.275,51.582) (69.610,51.229) (69.945,50.915) (70.280,50.636) (70.615,50.388) (70.950,50.166) (71.285,49.967) (71.620,49.789) (71.955,49.629) (72.290,49.485) (72.625,49.355) (72.960,49.237) (73.295,49.131) (73.630,49.035) (73.965,48.947) (74.300,48.868) (74.635,48.795) (74.970,48.729) (75.305,48.669) (75.640,48.614) (75.975,48.564) (76.310,48.518) (76.645,48.476) (76.980,48.438) (77.315,48.403) (77.650,48.371) (77.985,48.341) (78.320,48.314) (78.655,48.289) (78.990,48.266) (79.325,48.245) (79.660,48.226) (79.995,48.208) (80.330,48.192) (80.665,48.177) (81.000,48.163)};
\draw[tradeoffred,densely dashed,line width=0.9pt] plot coordinates {(14.000,41.891) (14.335,41.835) (14.670,41.782) (15.005,41.730) (15.340,41.681) (15.675,41.635) (16.010,41.590) (16.345,41.548) (16.680,41.507) (17.015,41.469) (17.350,41.432) (17.685,41.397) (18.020,41.364) (18.355,41.333) (18.690,41.304) (19.025,41.276) (19.360,41.249) (19.695,41.224) (20.030,41.201) (20.365,41.179) (20.700,41.158) (21.035,41.138) (21.370,41.120) (21.705,41.102) (22.040,41.086) (22.375,41.071) (22.710,41.056) (23.045,41.043) (23.380,41.030) (23.715,41.019) (24.050,41.008) (24.385,40.998) (24.720,40.988) (25.055,40.979) (25.390,40.971) (25.725,40.963) (26.060,40.956) (26.395,40.950) (26.730,40.943) (27.065,40.938) (27.400,40.932) (27.735,40.927) (28.070,40.923) (28.405,40.919) (28.740,40.915) (29.075,40.911) (29.410,40.908) (29.745,40.904) (30.080,40.902) (30.415,40.899) (30.750,40.896) (31.085,40.894) (31.420,40.892) (31.755,40.890) (32.090,40.888) (32.425,40.887) (32.760,40.885) (33.095,40.884) (33.430,40.882) (33.765,40.881) (34.100,40.880) (34.435,40.879) (34.770,40.878) (35.105,40.877) (35.440,40.877) (35.775,40.876) (36.110,40.875) (36.445,40.875) (36.780,40.874) (37.115,40.874) (37.450,40.874) (37.785,40.873) (38.120,40.873) (38.455,40.873) (38.790,40.873) (39.125,40.873) (39.460,40.873) (39.795,40.872) (40.130,40.873) (40.465,40.873) (40.800,40.873) (41.135,40.873) (41.470,40.873) (41.805,40.873) (42.140,40.874) (42.475,40.874) (42.810,40.875) (43.145,40.875) (43.480,40.876) (43.815,40.876) (44.150,40.877) (44.485,40.878) (44.820,40.879) (45.155,40.880) (45.490,40.881) (45.825,40.882) (46.160,40.884) (46.495,40.885) (46.830,40.887) (47.165,40.889) (47.500,40.892) (47.835,40.894) (48.170,40.897) (48.505,40.900) (48.840,40.904) (49.175,40.908) (49.510,40.913) (49.845,40.918) (50.180,40.924) (50.515,40.931) (50.850,40.939) (51.185,40.948) (51.520,40.958) (51.855,40.970) (52.190,40.984) (52.525,41.000) (52.860,41.018) (53.195,41.038) (53.530,41.062) (53.865,41.090) (54.200,41.121) (54.535,41.157) (54.870,41.198) (55.205,41.245) (55.540,41.299) (55.875,41.361) (56.210,41.431) (56.545,41.510) (56.880,41.601) (57.215,41.704) (57.550,41.821) (57.885,41.953) (58.220,42.103) (58.555,42.274) (58.890,42.467) (59.225,42.687) (59.560,42.937) (59.895,43.222) (60.230,43.548) (60.565,43.920) (60.900,44.349) (61.235,44.842) (61.570,45.412) (61.905,46.072) (62.240,46.837) (62.575,47.725) (62.910,48.751) (63.245,49.927) (63.580,51.249) (63.915,52.688) (64.250,54.170) (64.585,55.567) (64.920,56.707) (65.255,57.436) (65.590,57.683) (65.925,57.495) (66.260,56.996) (66.595,56.323) (66.930,55.584) (67.265,54.849) (67.600,54.155) (67.935,53.521) (68.270,52.950) (68.605,52.442) (68.940,51.990) (69.275,51.589) (69.610,51.235) (69.945,50.920) (70.280,50.640) (70.615,50.391) (70.950,50.168) (71.285,49.969) (71.620,49.790) (71.955,49.630) (72.290,49.486) (72.625,49.355) (72.960,49.238) (73.295,49.131) (73.630,49.035) (73.965,48.947) (74.300,48.868) (74.635,48.795) (74.970,48.730) (75.305,48.669) (75.640,48.614) (75.975,48.564) (76.310,48.519) (76.645,48.477) (76.980,48.438) (77.315,48.403) (77.650,48.371) (77.985,48.341) (78.320,48.314) (78.655,48.289) (78.990,48.266) (79.325,48.245) (79.660,48.226) (79.995,48.208) (80.330,48.192) (80.665,48.177) (81.000,48.163)};
\end{scope}
\begin{scope}
\clip (14,12) rectangle (81,32);
\draw[gridblue,line width=1pt] (14,19.5) -- (81,19.5);
\draw[tradeoffred,densely dashed,line width=0.95pt] plot coordinates {(14.000,25.134) (14.335,25.152) (14.670,25.168) (15.005,25.183) (15.340,25.197) (15.675,25.210) (16.010,25.223) (16.345,25.234) (16.680,25.244) (17.015,25.254) (17.350,25.263) (17.685,25.271) (18.020,25.279) (18.355,25.286) (18.690,25.293) (19.025,25.299) (19.360,25.305) (19.695,25.310) (20.030,25.315) (20.365,25.319) (20.700,25.323) (21.035,25.327) (21.370,25.331) (21.705,25.334) (22.040,25.337) (22.375,25.340) (22.710,25.343) (23.045,25.345) (23.380,25.348) (23.715,25.350) (24.050,25.352) (24.385,25.354) (24.720,25.355) (25.055,25.357) (25.390,25.359) (25.725,25.360) (26.060,25.362) (26.395,25.363) (26.730,25.364) (27.065,25.365) (27.400,25.367) (27.735,25.368) (28.070,25.369) (28.405,25.370) (28.740,25.371) (29.075,25.372) (29.410,25.373) (29.745,25.374) (30.080,25.375) (30.415,25.377) (30.750,25.378) (31.085,25.379) (31.420,25.380) (31.755,25.381) (32.090,25.383) (32.425,25.384) (32.760,25.386) (33.095,25.387) (33.430,25.389) (33.765,25.391) (34.100,25.392) (34.435,25.394) (34.770,25.397) (35.105,25.399) (35.440,25.401) (35.775,25.404) (36.110,25.407) (36.445,25.409) (36.780,25.413) (37.115,25.416) (37.450,25.420) (37.785,25.424) (38.120,25.428) (38.455,25.432) (38.790,25.437) (39.125,25.442) (39.460,25.448) (39.795,25.454) (40.130,25.460) (40.465,25.467) (40.800,25.474) (41.135,25.482) (41.470,25.490) (41.805,25.499) (42.140,25.509) (42.475,25.519) (42.810,25.529) (43.145,25.541) (43.480,25.552) (43.815,25.565) (44.150,25.578) (44.485,25.593) (44.820,25.607) (45.155,25.623) (45.490,25.640) (45.825,25.657) (46.160,25.675) (46.495,25.694) (46.830,25.713) (47.165,25.734) (47.500,25.755) (47.835,25.778) (48.170,25.801) (48.505,25.825) (48.840,25.850) (49.175,25.875) (49.510,25.902) (49.845,25.930) (50.180,25.958) (50.515,25.988) (50.850,26.018) (51.185,26.050) (51.520,26.083) (51.855,26.117) (52.190,26.152) (52.525,26.189) (52.860,26.227) (53.195,26.267) (53.530,26.308) (53.865,26.352) (54.200,26.398) (54.535,26.447) (54.870,26.498) (55.205,26.553) (55.540,26.611) (55.875,26.673) (56.210,26.739) (56.545,26.810) (56.880,26.886) (57.215,26.969) (57.550,27.058) (57.885,27.154) (58.220,27.259) (58.555,27.372) (58.890,27.496) (59.225,27.631) (59.560,27.777) (59.895,27.937) (60.230,28.110) (60.565,28.298) (60.900,28.500) (61.235,28.716) (61.570,28.942) (61.905,29.172) (62.240,29.396) (62.575,29.594) (62.910,29.736) (63.245,29.769) (63.580,29.617) (63.915,29.171) (64.250,28.302) (64.585,26.896) (64.920,24.937) (65.255,22.575) (65.590,20.133) (65.925,17.976) (66.260,16.358) (66.595,15.342) (66.930,14.844) (67.265,14.718) (67.600,14.826) (67.935,15.064) (68.270,15.365) (68.605,15.686) (68.940,16.005) (69.275,16.309) (69.610,16.592) (69.945,16.853) (70.280,17.092) (70.615,17.308) (70.950,17.504) (71.285,17.682) (71.620,17.843) (71.955,17.988) (72.290,18.120) (72.625,18.239) (72.960,18.347) (73.295,18.445) (73.630,18.534) (73.965,18.615) (74.300,18.689) (74.635,18.757) (74.970,18.818) (75.305,18.874) (75.640,18.925) (75.975,18.972) (76.310,19.015) (76.645,19.054) (76.980,19.090) (77.315,19.123) (77.650,19.153) (77.985,19.180) (78.320,19.206) (78.655,19.229) (78.990,19.251) (79.325,19.270) (79.660,19.288) (79.995,19.305) (80.330,19.320) (80.665,19.334) (81.000,19.347)};
\end{scope}
\node[text=gridblue,align=center,fill=white,inner sep=0.3mm] at (30,46.5) {exact model\\$D_v^0=D_v$};
\draw[gridblue,-{Stealth[length=1.2mm]}] (37,44.5) -- (42.46,40.9);
\node[text=tradeoffred,align=center,fill=white,inner sep=0.3mm] at (53,51.5) {30\% model error\\$D_v^0=1.3D_v$};
\draw[tradeoffred,-{Stealth[length=1.2mm]}] (60,49.8) -- (58.34,42.2);
\node[text=black!75,fill=white,inner sep=0.3mm] at (66,62.5) {amplification};
\draw[black!75,-{Stealth[length=1.2mm]}] (65,61) -- (65.61,57.5);
\node[text=gridblue,fill=white,inner sep=0.3mm] at (35,16) {preserved: $D_v^0=D_v$};
\draw[gridblue,-{Stealth[length=1.2mm]}] (47,17.3) -- (51,19.5);
\node[text=tradeoffred,fill=white,inner sep=0.3mm] at (39,29.3) {distorted: $D_v^0=1.3D_v$};
\draw[tradeoffred,-{Stealth[length=1.2mm]}] (48,27.8) -- (51,26.4);
\node at (47,36) {(a)};
\node at (47,4) {Frequency (Hz)};
\node at (47,-1) {(b)};
\end{tikzpicture}
\caption{Load-model mismatch for $D_v=0.7$, $G_p=-0.15/(1+0.025s)$, $G_g=1/(1+0.025s)$, and $u=F(p-D_v^0v)$, with $F$ as in Fig.~\ref{fig:bode_area}. Here $D_v^0$ is the nominal compensation parameter. (a) Server-demand leakage. (b) Grid-response amplitude ratio. Both are normalized by their uncontrolled responses. Shaded bands: 1--4 and 8--20 Hz.}
\label{fig:cross_band_limit}
\end{figure}

\subsection{Implications for Controller Design}

The preceding limits require joint design of measurements and feedback dynamics to distinguish sources, allocate performance across frequency bands, and maintain grid-response accuracy under model variation. The major implications for the controller designs are: 

First, select measurements with distinct sources. The source signatures should satisfy \eqref{eq:rank_condition} and remain well conditioned across frequencies and operating points. A power/voltage or power/frequency subset is sufficient when it meets these requirements. Second, specify isolation depth and bandwidth together with the grid-response band and acceptable out-of-band amplification. Account for the integral constraint \eqref{eq:bode_general}, including actuator and measurement dynamics. Constrain sensitivity peaks outside the isolation band, particularly where accurate grid support is required. Third, jointly minimize server-demand leakage $T_w$ and grid-response error $T_g-T_g^\star$ over different load and grid conditions. Since the actual $S_u$ multiplies the matching residual in \eqref{eq:grid_tracking_mismatch}, address model mismatch and sensitivity peaks together.

\section{Source-Selective Controller Synthesis}
\label{sec:controller}

This section designs a controller for the isolation and grid-response objectives in \eqref{eq:two_objectives}. Measurement selection follows the source-rank condition \eqref{eq:rank_condition}, while isolation bandwidth and allowable sensitivity peaks account for the integral constraint \eqref{eq:bode_general}. The design also addresses the model-mismatch effect in \eqref{eq:grid_tracking_mismatch} across operating points.

\subsection{Measurement Residual and Physical Design Requirements}
We first derive the controller for the nominal plant. In \eqref{eq:port_measurement}, the nominal HESS contribution to the measurements is $\bm{h}_u^0u$. Subtracting it gives $\bm{r}_y=\bm{y}-\bm{h}_u^0u=\bm{h}_w^0w+\bm{h}_g^0g$ when the plant matches the model. The residual-feedback control law is
\begin{equation}
u=-\bm{q}^{\mathsf T}(s)\bm{r}_y = -\bm{q}^{\mathsf T}(s)(\bm{y}-\bm{h}_u^0u).
 \label{eq:controller}
\end{equation}
Here $\bm{q}(s)$ contains the controller transfers from the measurement residuals to HESS power. It determines how the power, voltage, and frequency residuals drive HESS charging or discharging to reduce server-demand-induced grid import while tracking the specified grid response. With a stable proper $\bm{q}$ and a well-posed realization, this is a stable-plant internal-model-control construction~\cite{garcia1982imc}. Equation~\eqref{eq:controller} gives the equivalent feedback in Fig.~\ref{fig:framework}, $\bm{K}=\bm{q}/(1-\bm{q}^{\mathsf T}\bm{h}_u^0)$. With an exact nominal model, substitution into \eqref{eq:general_feedback_maps} and \eqref{eq:return_sensitivity} gives
\begin{equation}
 \begin{aligned}
 T_i^0=P_i^0-P_u^0\bm{q}^{\mathsf T}\bm{h}_i^0,\quad i\in\{w,g\},\quad S_u^0=1-\bm{q}^{\mathsf T}\bm{h}_u^0.
 \end{aligned}
 \label{eq:closed_loop_maps}
\end{equation}

We design $\bm q$ to suppress server-induced PCC-power fluctuations while delivering the grid response $T_g^\star$. To account for the cross-band limitation in \eqref{eq:bode_general}, we also constrain out-of-band amplification. These objectives must hold at each considered operating point:
\begin{equation}
	\begin{alignedat}{2}
		\text{Isolation:}\quad &|T_w/P_w|\le\epsilon_w, &\quad&\omega\in B,\\
		\text{Grid response:}\quad &|T_g-T_g^\star|\le\epsilon_g, &\quad&\omega\in\Omega_g,\\
		\text{Peak limit:}\quad &|S_u|\le M, &\quad&\omega\in[0,\Omega_c]\setminus B.
	\end{alignedat}
	\label{eq:design_requirements}
\end{equation}
Here $B\subset[0,\Omega_c]$ is the server-demand isolation band, $\Omega_g\subset[0,\Omega_c]$ is the required grid-response band, and $\Omega_c$ is the upper angular frequency of the finite-band assessment. For example, a 1--20-Hz isolation requirement corresponds to $B=2\pi[1,20]$ rad/s; $B$ may also comprise disjoint intervals. The grid-response band is chosen according to the support target $T_g^\star$ and need not coincide with $B$; $\Omega_c$ must extend beyond the sensitivity peaks being assessed. All transfers are evaluated at $s=j\omega$, and $P_w\ne0$ is assumed on $B$. The first row bounds PCC-power ripple relative to the uncontrolled response to the same server-demand input, so $\epsilon_w$ is dimensionless. The second row bounds the complex grid-response error, including both magnitude and phase. Since $S_u=T_w/P_w$ wherever $P_w\ne0$, the third row uses $M$ to limit amplification relative to the uncontrolled server-demand response outside $B$ and below $\Omega_c$.

\subsection{Model Variation}
Let $\rho$ describe operating conditions, including average server loading and grid strength. Keep the implemented nominal model $\bm{h}_u^0$ fixed and define $\delta\bm{h}_u=\bm{h}_u(\rho)-\bm{h}_u^0$, including both modeling error and operating-point variation. The residual is now $\bm{r}_y=\bm{y}-\bm{h}_u^0u=\bm{h}_ww+\bm{h}_gg+\delta\bm{h}_uu$. The scalar $\eta=\bm{q}^{\mathsf T}\delta\bm{h}_u$ is the unintended feedback remaining after self-response subtraction. The controller remains $\bm{K}=\bm{q}/(1-\bm{q}^{\mathsf T}\bm{h}_u^0)$. Substitution into \eqref{eq:controller}, \eqref{eq:port_power}, and \eqref{eq:return_sensitivity} gives the actual responses
\begin{equation}
	\begin{aligned}
		D&=1+\eta,\qquad T_i-T_i^\star=N_i/D,\\
		N_i&=(P_i-T_i^\star)D-P_u\bm q^{\mathsf T}\bm h_i,
		\quad i\in\{w,g\},
	\end{aligned}
	\label{eq:error_numerators}
\end{equation}
where $T_w^\star=0$. The numerator $N_i$ collects the uncompensated terms in each tracking error; $D$ captures feedback left by imperfect self-response subtraction. Since $S_u=(1-\bm q^{\mathsf T}\bm h_u^0)/D$, this is also the mismatch-amplification mechanism of Section~III: reducing the numerators alone is insufficient.

Equation~\eqref{eq:error_numerators} also exposes a stability issue: stable controller filters do not prevent mismatch feedback from creating unstable zeros of $D$. Internal stability of the plant with incremental HESS power $u=0$ is checked at each design operating point before synthesis. A reference controller $\bm q_b$ then defines $D_b=1+\bm q_b^{\mathsf T}\delta\bm h_u$; $\bm q_b=\bm0$ gives the stable reference $D_b=1$, while the cases use a verified stable engineering controller. We constrain $D/D_b$ to preserve stability while optimizing isolation and support.

\begin{proposition}[Reference-centered stability certificate]
\label{prop:centered_stability}
Assume the plant with $u=0$ is internally stable, $\bm h_u^0$ and $\bm q,\bm q_b$ are stable proper, both residual interconnections are well posed, and $D_b^{-1}\in\RHinf$. Let $a=D/D_b$ be proper with $a(\infty)=1$. If a constant $\delta_0>0$ satisfies
\begin{equation}
 \operatorname{Re}a(j\omega,\rho)\ge\delta_0
 \quad\text{for all }\omega\in\mathbb R
 \label{eq:residual_robustness}
\end{equation}
at each considered $\rho$, then $D^{-1}\in\RHinf$ and the residual feedback is internally stable.
\end{proposition}
\begin{proof}
The stable factors give $a\in\RHinf$. The boundary condition and $a(\infty)=1$ imply $\operatorname{Re}a(s)\ge\min(\delta_0,1)>0$ throughout the closed right half-plane by the harmonic minimum principle. Hence $a^{-1}$ is stable proper, and so is $D^{-1}=a^{-1}D_b^{-1}$.
\end{proof}

We now enforce \eqref{eq:design_requirements}. At each sampled frequency and operating point, use separate margins $\delta_y,\delta_u\ge\delta_0>0$ for performance and measurement gain, respectively. Either bounds $|D/D_b|$ from below; in particular, $|N_i/D_b|/\delta_y$ bounds the actual error.

\subsection{Joint Performance Design}

To select an implementable controller, fix stable proper basis functions $b_\ell$ and parameterize
\begin{equation}
	\bm q(s)=F_b(s)\sum\nolimits_{\ell=1}^{N}\bm\theta_\ell b_\ell(s),
	\label{eq:basis}
\end{equation}
where $F_b$ represents HESS power tracking. The coefficients map residuals to the command $u_c=-[\sum_\ell\bm\theta_\ell b_\ell]^{\mathsf T}\bm r_y$, and $u=F_bu_c$ includes actuator bandwidth. This frequency-shaping parameterization~\cite{feng2026participation} makes $D$ and $N_i$ affine in the coefficients. We select feasible coefficients using
\begin{equation}
 \min_{\{\bm\theta_\ell\},\,\{\delta_y,\delta_u\}}
 \sum_{(\omega,\rho)\in\mathcal D}
 W_w|N_w|^2+W_g|N_g|^2+
 W_u\|\bm q^{\mathsf T}\bm\Sigma\|_2^2,
 \label{eq:qp}
\end{equation}
where $\mathcal D$ samples frequencies and operating points, and $\bm\Sigma$ scales $(p,v,f)$ perturbations. Fixed nonnegative weights include quadrature and averaging; $W_w$ includes $\max(|P_w|,10^{-4})^{-2}$ for isolation normalization. The last term penalizes measurement-to-HESS gain.

The constraints are
\begin{subequations}
 \label{eq:robust_performance_bounds}
 \begin{align}
 \operatorname{Re}(D/D_b)&\ge\max\{\delta_y,\delta_u\},\label{eq:constraint_stability}\\*
 |N_w/D_b|&\le\delta_y\epsilon_w|P_w|,\label{eq:constraint_isolation}\\*
 |N_g/D_b|&\le\delta_y\epsilon_g,\label{eq:constraint_grid}\\*
 |(1-\bm q^{\mathsf T}\bm h_u^0)/D_b|&\le\delta_y M,\label{eq:constraint_peak}\\*
 \nu_b+\|(\bm q-\bm q_b)^{\mathsf T}\bm\Sigma/D_b\|_2
 &\le\delta_u\nu .\label{eq:constraint_gain}
 \end{align}
\end{subequations}
Here $\bm q_b$ and $D_b$ are fixed; optimizing the coefficients changes $\bm q$, $D$, and $N_i$. Constraint~\eqref{eq:constraint_stability} gives $|D|\ge\delta_y|D_b|$ and $|D|\ge\delta_u|D_b|$; when enforced at all frequencies, it ensures stability under Proposition~\ref{prop:centered_stability}. The next three rows impose isolation, grid tracking, and out-of-band sensitivity limits through $\delta_y$; for example, $|T_g-T_g^\star|\le|N_g/D_b|/\delta_y\le\epsilon_g$. The last row uses $\delta_u$ and the triangle inequality to bound scaled measurement-to-HESS gain, with $\nu_b=\|\bm q_b^{\mathsf T}\bm\Sigma/D_b\|_2$. Stability and gain constraints apply at all samples, while $\epsilon_g$ may vary with frequency. Sampled solutions require separate pole and peak checks; $D_b$ provides the fixed normalization.

Fixing $D_b$ makes the two real-part inequalities affine and the norm bounds conic, yielding an SOCP with a convex quadratic objective~\cite{boyd2004convex}; $\delta_y$ and $\delta_u$ are margins selected during offline design. For transient-only support, a washout enforces $\bm q(0)=\bm0$ and zero steady-state incremental HESS power, requiring $T_g^\star(0)=P_g(0)$ for exact tracking.

\section{Case Studies}
\label{sec:case_studies}

\subsection{Settings, Targets, and Metrics}
\label{sec:case_settings}

\subsubsection{Physical systems and operating points}
Fig.~\ref{fig:case_topologies} shows two test systems: a single-generator interconnection and a modified IEEE nine-bus network. Both retain the released server/converter and grid-device simulation models of~\cite{zhao2026datacenter}, including AFE, UPS, VSI, and PSU dynamics, on 100-MVA, 60-Hz bases. Modified cases and reproduction scripts are available online.\footnote{\url{https://github.com/VictorCFeng/Data-Center-Oscillation}} In the single-generator case, a finite synchronous machine ($H=3$ s, native excitation/governor) supplies a 50-MW data center through $0.02+j0.19$ pu impedance; the voltage is initialized at 1.04 pu and evolves with the machine dynamics. The nine-bus case retains the released branches, with synchronous, grid-forming, and grid-following units at Buses~1--3 and 62.5/50-MW data centers at Buses~5/8. Remaining loads are passive impedances. Each PCC has a HESS with $F_b=200/(s+200)$ and a $\pm20$-MW aggregate output limit. Controllers use local power, voltage, and AFE-PLL frequency estimates, retaining network and PLL dynamics. Loading and line-impedance multipliers $(\ell,z)$ define thirteen separately initialized equilibria (Table~\ref{tab:case_settings}). Previously inspected points 6--9 are used for evaluation and excluded from coefficient fitting. Nonlinear time-domain simulations are performed in MATLAB R2023b using \texttt{ode15s}. Controller synthesis is implemented in Python using CVXPY~\cite{diamond2016cvxpy} with the Clarabel solver~\cite{goulart2024clarabel}.

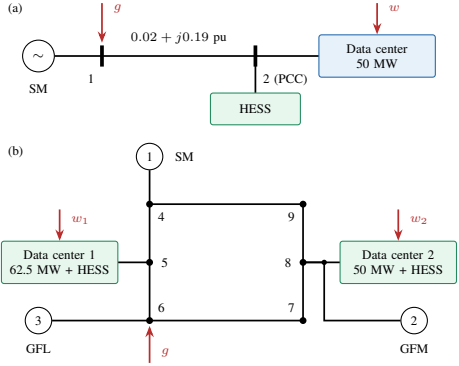
\begin{figure}[!t]
\centering
\begin{tikzpicture}[x=1mm,y=1mm,scale=0.7,transform shape,font=\scriptsize,>=Stealth,
  wire/.style={line width=0.65pt},
  bus/.style={line width=1.4pt},
  device/.style={draw,rounded corners=1pt,align=center,inner sep=1.2mm,fill=lightblue,draw=gridblue},
  storage/.style={device,fill=lightgreen,draw=supportgreen},
  event/.style={draw=tradeoffred,-{Stealth[length=1.5mm]},line width=0.65pt}]
\path[use as bounding box] (0,0) rectangle (85,71);
\node[anchor=west] at (0,68) {(a)};
\node[draw,circle,minimum size=6mm,inner sep=0pt] (smA) at (7,59) {$\sim$};
\node[below=1mm] at (smA.south) {SM};
\draw[wire] (smA.east) -- (19,59) -- (48,59) -- (60,59);
\draw[bus] (19,57) -- (19,61);
\draw[bus] (48,57) -- (48,61);
\node[below left=0.3mm] at (19,57) {1};
\node[below right=0.3mm] at (48,57) {2 (PCC)};
\node[above=1mm] at (33.5,59) {$0.02+j0.19$ pu};
\node[device,minimum width=22mm,minimum height=8mm] (dcA) at (71,59) {Data center\\50 MW};
\node[storage,minimum width=17mm,minimum height=6mm] (hessA) at (48,49) {HESS};
\draw[wire] (48,59) -- (hessA.north);
\draw[event] (19,69) -- (19,61.5);
\node[anchor=west,text=tradeoffred] at (20,68) {$g$};
\draw[event] (71,69) -- (dcA.north);
\node[anchor=west,text=tradeoffred] at (72,68) {$w$};
\node[anchor=west] at (0,41) {(b)};
\coordinate (b4) at (28,31);
\coordinate (b5) at (28,20);
\coordinate (b6) at (28,9);
\coordinate (b7) at (57,9);
\coordinate (b8) at (57,20);
\coordinate (b9) at (57,31);
\draw[wire] (b4) -- (b5) -- (b6) -- (b7) -- (b8) -- (b9) -- (b4);
\foreach \n in {4,5,6,7,8,9} {\fill (b\n) circle (0.65mm);}
\node[below right=0.5mm] at (b4) {4};
\node[right=1mm] at (b5) {5};
\node[above right=0.5mm] at (b6) {6};
\node[above left=0.5mm] at (b7) {7};
\node[left=1mm] at (b8) {8};
\node[below left=0.5mm] at (b9) {9};
\node[draw,circle,minimum size=5mm,inner sep=0pt] (smB) at (28,40) {1};
\draw[wire] (smB.south) -- (b4);
\node[right=1mm] at (smB.east) {SM};
\node[draw,circle,minimum size=5mm,inner sep=0pt] (gfl) at (7,9) {3};
\draw[wire] (gfl.east) -- (b6);
\node[below=0.7mm] at (gfl.south) {GFL};
\node[draw,circle,minimum size=5mm,inner sep=0pt] (gfm) at (78,9) {2};
\draw[wire] (b8) -- (61,20) -- (61,9) -- (gfm.west);
\node[below=0.7mm] at (gfm.south) {GFM};
\node[storage,minimum width=22mm,minimum height=8mm] (site1) at (11,20) {Data center 1\\62.5 MW + HESS};
\draw[wire] (site1.east) -- (b5);
\node[storage,minimum width=22mm,minimum height=8mm] (site2) at (75,20) {Data center 2\\50 MW + HESS};
\draw[wire] (b8) -- (site2.west);
\fill (61,20) circle (0.5mm);
\draw[event] (11,30) -- (site1.north);
\node[anchor=west,text=tradeoffred] at (12,28) {$w_1$};
\draw[event] (75,30) -- (site2.north);
\node[anchor=west,text=tradeoffred] at (76,28) {$w_2$};
\draw[event] (28,1) -- (28,8);
\node[anchor=west,text=tradeoffred] at (29,3) {$g$};
\end{tikzpicture}
\caption{Test-system topologies. (a) Single-generator interconnection with one data center. (b) Two data centers in the modified IEEE nine-bus network. Numbers identify buses; $w_i$ enters server demand and $g$ is the external-grid disturbance. Each data center has a PCC-connected HESS; passive shunt loads are omitted.}
\label{fig:case_topologies}
\end{figure}

\begin{table}[!t]
\caption{Operating points and common requirements}
\label{tab:case_settings}
\centering\footnotesize
\begin{tabular}{@{}p{0.21\columnwidth}p{0.74\columnwidth}@{}}
\toprule
Item & Setting\\
\midrule
Points 1--5 & $(1,1),(0.8,1),(1.2,1),(1,0.8),(1,1.2)$\\
Points 6--9 & $(0.9,0.9),(0.9,1.1),(1.1,0.9),(1.1,1.1)$\\
Points 10--13 & $(0.8,0.8),(0.8,1.2),(1.2,0.8),(1.2,1.2)$\\
Design points & 1--5 and 10--13\\
Isolation & 1--20 Hz; leakage $\le$0.65\\
Grid service & 0.005--2 / 0.005--5.5 Hz; error$\le$0.35/0.45 (one/two data centers)\\
Out-of-band & Leakage $\le$2.8 outside 1--20 Hz, over 0.001--200 Hz\\
Meas. scales $\bm\Sigma$ & $(0.003\text{ pu},0.001\text{ pu},0.01\text{ Hz})$ per site's $(p,v,f)$\\
Gain limit & 1.10 MW (one data center); $\sqrt{2}$ MW (two data centers)\\
\bottomrule
\end{tabular}
\end{table}

\subsubsection{Required grid service}
Here $T_g^\star$ is the closed-loop response under the reference storage law $u_i=-H_{f,i}^{\rm ref}f_i$ ($u_i$ in pu, $f_i$ in Hz), where $H_{f,i}^{\rm ref}=F_b\bigl[0.2\frac{s}{s+0.1}\frac{4\pi}{s+4\pi}+k_i\frac{s}{s+\pi}\frac{24\pi}{s+24\pi}\frac{1+s/(2\pi z_i)}{1+s/(2\pi p_i)}\bigr]$.
The first term combines 0.2-pu/Hz transient droop, a 10-s washout, and a 2-Hz low-pass. The damping branch combines a 0.5--12-Hz band-pass and lead-lag gain $k_i$, zero $z_i$, and pole $p_i$ (Hz); $F_b$ models HESS tracking. The single-generator case uses droop alone ($k_1=0$). For the network, a grid-only search over points 1--5 tunes the coefficients for frequency/voltage response and modal damping, giving $(k_i,z_i,p_i)=(0.919112,11.683569,14.142877)$ at Bus~5 and $(-0.773512,12,2.012619)$ at Bus~8. Nominal damping ratios near 2.25/3.89 Hz increase from 5.91\%/15.06\% to 10.71\%/18.71\%. With $w=0$, \eqref{eq:general_feedback_maps} gives $T_g^\star$ using $K_p=K_v=0$, $K_f=H_f^{\rm ref}$, and both reference loops closed for the network. The reference coefficients remain fixed across operating points.

\subsubsection{Disturbances and evaluation}
Three 30-s tests apply server fluctuations alone, a grid-load disturbance alone, and both. Server input over 2--26 s combines a 2.5-MW linear 2--11.6-Hz chirp and a 1.2-MW, 13-Hz sinusoid with 0.2-s ramps. The sites' chirp phases differ by $60^\circ$; their 13-Hz components are in phase. This synthetic excitation tests isolation under frequency variation. At $t=5$ s, a resistive load is switched on at Bus~1 (single-generator case) or Bus~6 (nine-bus case), drawing $5|V(t)/V(0)|^2$ MW. The load is rated at 5 MW at the initial bus voltage. Outputs are sampled every 4 ms, refined to 0.1 ms over 5--5.25 s for the grid transient.

Define server leakage $A=\max_j A_j$, with $A_j=\|\bm T_{w,j}\|_2/\max(\|\bm P_{w,j}\|_2,10^{-4})$, and grid error $E_g=\|\bm T_g-\bm T_g^\star\|_2/\max(\|\bm T_g^\star\|_2,0.05)$. The columns collect responses at all PCCs, avoiding cancellation between sites. The grid-error requirement is $E_g\le\bar\epsilon_g$, with $\bar\epsilon_g$ from Table~\ref{tab:case_settings}. Measurement-to-HESS gain is the peak singular value after $\bm\Sigma$ scaling, in MW. Time-domain error is measured relative to zero for the server-only event and to the same-point nonlinear grid-only reference otherwise. RMS is the square root of the time integral of the summed squared PCC errors divided by the event-window duration: 2--26 s for server excitation, 5--30 s for the grid step, and 2--30 s for the combined event. The RMS evaluation includes the initial grid transient. All methods receive identical inputs at each equilibrium.

\subsection{Comparators and Affine Controller Design}
\label{sec:case_controllers}

Seven methods are compared; plots omit the preliminary Lead + support design:
\begin{itemize}[\IEEEnocalcleftmargintrue\setlength{\leftmargin}{0pt}\setlength{\itemindent}{1em}\setlength{\labelwidth}{0.5em}\setlength{\labelsep}{0.5em}]
\setlength{\itemsep}{0pt}
\setlength{\parsep}{0pt}
\item \emph{No HESS}: native plant controls with $u=0$.
\item \emph{Power only}: a power-feedback comparator inspired by the smoothing objective of~\cite{lyu2026pcc}. It uses $p$ as its sole feedback signal and is tuned for isolation.
\item \emph{RBC (adapted)}: the rule-based baseline from~\cite[Sec.~3.1]{you2026hess}, applied to $d=p+u$ with complementary 0.5/8-Hz allocation. We replace the original 0.25/0.015-s actuators, discrete filters, and future-horizon mean with common 5-ms actuators, continuous filters, and an equilibrium schedule.
\item \emph{RBC + support}: adapted RBC plus our frozen frequency-support branch, an addition to~\cite{you2026hess}.
\item \emph{Lead + support}: engineering washout and lead/low-pass compensation of $d$ plus frequency support, tuned at points 1--5.
\item \emph{Lead + support (opt.)}: the same structure retuned at nine design points under Table~\ref{tab:case_settings}, supplying the local reference $\bm q_b$. Search weights are 1/1.5 for log-band leakage/grid error and 0.02 for scaled measurement-to-HESS gain.
\item \emph{Proposed}: SOCP~\eqref{eq:qp} subject to~\eqref{eq:robust_performance_bounds}, designed independently at each data center and then connected.
\end{itemize}

The proposed controller uses $\bm h_u^0=[-1,0,0]^{\mathsf T}$, giving $\bm r_y=[p+u,v,f]^{\mathsf T}$. The corresponding local residual-feedback form of Lead + support (opt.) defines $\bm q_b$. The residual is computed directly from measurements, with self-response approximation error represented by $D$. Each local channel uses $F_bW_h$ times the basis $1$, $a/(s+a)$ for the rates in Table~\ref{tab:case_basis}, and $[a_c/(s+a_c)]^2$, where $W_h=s/(s+2\pi\,0.002)$ and $a_c$ is the last listed rate. The proposed controllers have 6 and 17 states, including actuator dynamics; the engineering comparators have the same orders after minimal realization. The affine designs use 15/45 free coefficients and $(d,v,f)$ feedback; the engineering comparators use 4/5 tuned parameters and $(d,f)$ feedback.

\begin{table}[!t]
\caption{Fixed basis and implemented state budgets}
\label{tab:case_basis}
\centering\footnotesize
\begin{tabular}{@{}lp{0.65\columnwidth}r@{}}
\toprule
Site & Pole rates $a$ (rad/s, rounded) & States\\
\midrule
Single & 0.1, 12.5664, 4398.2297 & 6\\
Bus 5 & 0.1, 12.5664, 3.1416, 75.3982, 88.8623, 2693.2839 & 9\\
Bus 8 & 0.1, 12.5664, 3.1416, 12.6457, 2693.2839 & 8\\
\bottomrule
\end{tabular}
\end{table}

Local nine-bus designs retain the coupled network with the other HESS at $u=0$. Local targets close one reference controller; network evaluation closes both. The $\sqrt{2}$-MW gain ceiling applies to each local design and is checked after interconnection. In \eqref{eq:qp}, weights 1, 3, and 20 multiply normalized server-numerator energy over 1--20 Hz, grid-numerator energy normalized by local target step energy, and $\|\bm q^{\mathsf T}\bm\Sigma\|_2^2$, respectively, averaged over design points. Server and gain terms use log-frequency quadrature; the grid term uses $1/\omega^2$ weighting. Grid and gain terms cover 0.001--$10^4$ Hz. We use $\delta_0=10^{-5}$ and frequency-exchange SOCP, followed by interconnected pole and peak checks.

\subsection{Server-Oscillation Isolation}
\label{sec:case_isolation}

We first apply server disturbances alone ($g=0$) to test whether the controller prevents their propagation into the grid. In Fig.~\ref{fig:case_server_frequency}(a),(b), $A$ measures the remaining PCC-power oscillation relative to no HESS; smaller values mean stronger isolation. Over 1--20 Hz and thirteen operating points, the proposed design reduces the worst $A$ from 0.0581 to 0.0090 in the single-generator case and from 0.0371 to 0.0252 in the nine-bus case, compared with Lead + support (opt.) at the same state budget. Thus the remaining oscillation is at most 0.90\% and 2.52\% of the uncontrolled response, respectively. Panels (c),(d) show the resulting propagation to machine frequency.

\begin{figure}[!t]
\centering
\includegraphics[width=\columnwidth]{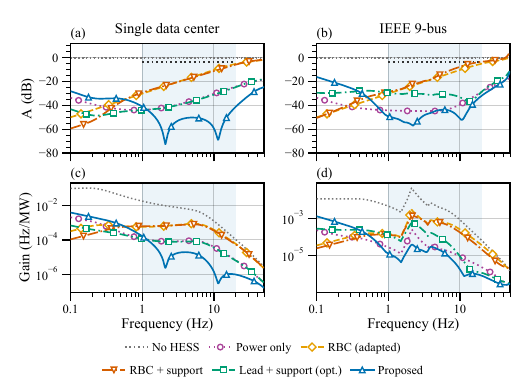}
\caption{Server-input frequency responses. Left: single-generator case; right: modified IEEE nine-bus case. (a),(b) Worst leakage $20\log_{10}A$ over thirteen equilibria. (c),(d) Machine-frequency gain, maximized over equilibria and individual server inputs. Shading marks 1--20 Hz; dotted segments in (a),(b) mark $A=0.65$.}
\label{fig:case_server_frequency}
\end{figure}

Fig.~\ref{fig:case_server_time} shows PCC ripple during the chirp-plus-tone test at point~13, with vertical zooms resolving controlled responses. Table~\ref{tab:case_time_results} reports full-event RMS. Across all operating points, the proposed worst RMS is 0.00625 MW in the single-generator case and 0.01711 MW in the nine-bus case, versus 0.04389 and 0.08831 MW for Lead + support (opt.): reductions of 85.8\% and 80.6\%. Isolation alone is not sufficient: power-only feedback also achieves a small nine-bus server RMS (0.02588 MW), yet its grid-event tracking error is 0.60792 MW, versus 0.00790 MW for the proposed controller. The next subsection therefore tests whether reduced server ripple is accompanied by the required grid support.

\begin{figure}[!t]
\centering
\includegraphics[width=\columnwidth]{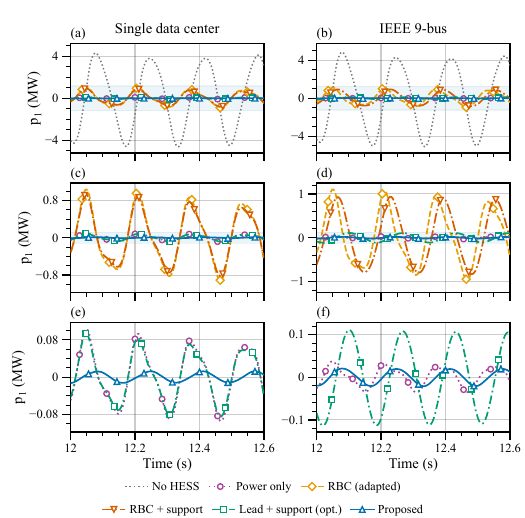}
\caption{Server-only PCC ripple at point~13, with a common 12--12.6-s window. Left: single-generator case; right: Bus~5 of the nine-bus case. Rows show all methods, controlled methods, and the three low-ripple methods (power only, optimized lead/support, and proposed). Shading marks the next row's vertical range. Full-event metrics appear in Table~\ref{tab:case_time_results}.}
\label{fig:case_server_time}
\end{figure}

\subsection{Grid-Response Tracking and Physical Support}
\label{sec:case_support}

Fig.~\ref{fig:case_grid_frequency} shows the magnitude of the complex grid-response error, including phase. Power-only feedback and unaugmented RBC suppress grid-related power variation, resulting in tracking errors near unity in the service band. In the nine-bus case, the proposed worst service-band error decreases from 0.4500 to 0.3691 relative to the optimized engineering comparator. In the single-generator case it increases from 0.0061 to 0.0094, while remaining far below 0.35. These results are summarized in Table~\ref{tab:case_results}.

\begin{figure}[!t]
\centering
\includegraphics[width=\columnwidth]{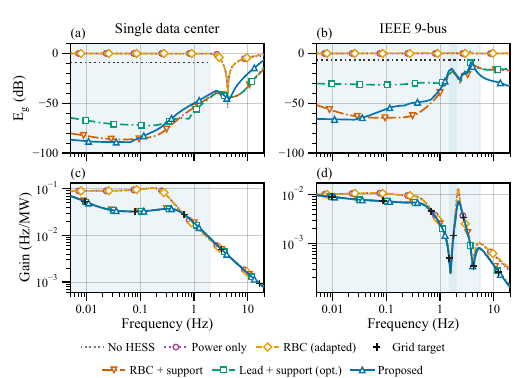}
\caption{Grid-input frequency responses. Left: single-generator case; right: modified IEEE nine-bus case. (a),(b) Worst complex target error $20\log_{10}E_g$ over thirteen equilibria; dotted segments mark the prescribed limits. (c),(d) Machine-frequency gain at point~13; crosses denote the grid-support reference. Shading marks the bands over which grid tracking is specified.}
\label{fig:case_grid_frequency}
\end{figure}

The lower panels of Fig.~\ref{fig:case_grid_frequency} compare the physical frequency response with the grid-support reference. Fig.~\ref{fig:case_grid_time} shows PCC exchange, its tracking error, and machine frequency for the grid step. Zero tracking error corresponds to the prescribed grid-support response. Logarithmic elapsed time resolves the initial transient, which is included in the RMS. Relative to the optimized engineering comparator, worst grid RMS decreases from 0.00715 to 0.00383 MW (46.5\%) in the single-generator case and from 0.03531 to 0.00790 MW (77.6\%) in the nine-bus case. RBC plus support retains a smaller single-generator grid RMS (0.00350 MW), but its nine-bus value (0.02067 MW) and server leakage are larger. Worst machine-frequency drops decrease from 0.5170 to 0.3285 Hz and from 0.05569 to 0.04854 Hz, respectively, relative to no HESS. Other support-capable comparators give similar frequency improvement. The proposed controller's main advantage is therefore better isolation while meeting the prescribed grid-service requirement.

\begin{figure}[!t]
\centering
\includegraphics[width=\columnwidth]{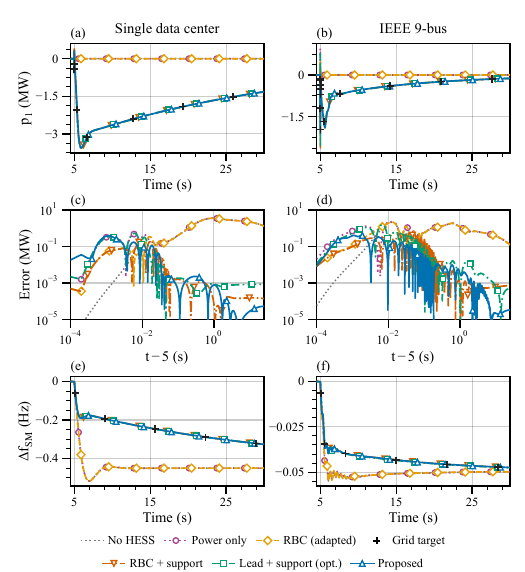}
\caption{Grid step at point~13. Left: single-generator case; right: nine-bus case (PCC power at Bus~5). Rows show PCC power, absolute PCC-target error, and machine-frequency deviation. Crosses denote the grid-support reference. The middle row uses elapsed time $t-5$ on a logarithmic axis; the displayed error range starts at $10^{-5}$ MW.}
\label{fig:case_grid_time}
\end{figure}

\subsection{Operating-Point Variation and Resource Cost}
\label{sec:case_robustness}

We now test the multi-operating-point design in Section~\ref{sec:controller} by applying both disturbances at thirteen loading and network-impedance conditions. Each test starts from its corresponding equilibrium with fixed controller coefficients.

Fig.~\ref{fig:case_operating_points}(a),(b) compares PCC power with the required grid-only response. Tracking this reference requires suppressing server fluctuations while delivering grid support. The proposed design achieves the lowest worst-case combined RMS among the compared methods (Table~\ref{tab:case_results}). Relative to Lead + support (opt.), RMS decreases from 0.04096 to 0.00668 MW in the single-generator case and from 0.08783 to 0.01760 MW in the nine-bus case: reductions of 83.7\% and 80.0\% at the same 6- and 17-state orders. Both methods meet the prescribed requirements at all thirteen points; the proposed design achieves greater tracking accuracy within these limits.

The literature-based RBC + support gives worst combined RMS values of 0.40025/0.67524 MW. The improvement over this comparator supports explicitly shaping both disturbance responses across operating points. In the nine-bus case, the independently designed local controllers retain this advantage when operating together in the coupled network.

\begin{figure}[!t]
\centering
\includegraphics[width=\columnwidth]{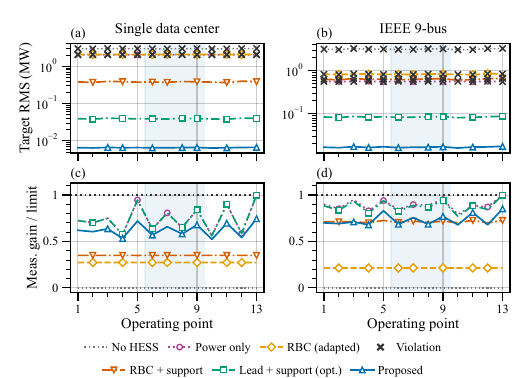}
\caption{Operating-point comparison. Left: single-generator case; right: modified IEEE nine-bus case. (a),(b) Combined-event PCC-target RMS over 2--30 s, including all PCC errors; crosses mark violation of any frequency-domain requirement. (c),(d) Scaled measurement-to-HESS gain divided by its limit; the dotted line is unity. Shading identifies evaluation points 6--9.}
\label{fig:case_operating_points}
\end{figure}

\begin{table}[!t]
\caption{State budget and worst frequency- and time-domain performance}
\label{tab:case_results}
\label{tab:case_time_results}
\centering\scriptsize
\setlength{\tabcolsep}{1.5pt}
\begin{tabular}{@{}lrrrrrrr@{}}
\toprule
Method & States & $\max A$ & $\max E_g$ & Pass & \multicolumn{3}{c}{RMS error (MW)}\\
\cmidrule(l){6-8}
 & & & & & Server & Grid & Combined\\
\midrule
\multicolumn{8}{c}{Single-generator interconnection with one data center}\\
No HESS & 0 & 1.0000 & 1.0000 & 0/13 & 2.40874 & 2.17384 & 3.02516\\
Power only & 4 & 0.0619 & 1.0000 & 0/13 & 0.04729 & 2.17384 & 2.05437\\
RBC (adapted) & 4 & 0.5536 & 1.0000 & 0/13 & 0.46800 & 2.17384 & 2.09398\\
RBC + support & 7 & 0.5146 & 0.0067 & 13/13 & 0.43221 & 0.00350 & 0.40025\\
Lead + support & 6 & 0.0653 & 0.0061 & 13/13 & 0.05052 & 0.00705 & 0.04704\\
Lead + support (opt.) & 6 & 0.0581 & 0.0061 & 13/13 & 0.04389 & 0.00715 & 0.04096\\
Proposed & 6 & 0.0090 & 0.0094 & 13/13 & 0.00625 & 0.00383 & 0.00668\\
\midrule
\multicolumn{8}{c}{Modified IEEE nine-bus, two data centers}\\
No HESS & 0 & 1.0000 & 1.2817 & 0/13 & 3.53852 & 0.60835 & 3.31994\\
Power only & 8 & 0.0370 & 1.0019 & 0/13 & 0.02588 & 0.60792 & 0.57413\\
RBC (adapted) & 8 & 0.5647 & 1.0457 & 0/13 & 0.69275 & 0.60764 & 0.85751\\
RBC + support & 20 & 0.6809 & 0.4948 & 1/13 & 0.72924 & 0.02067 & 0.67524\\
Lead + support & 17 & 0.0393 & 0.4340 & 12/13 & 0.10231 & 0.04012 & 0.10144\\
Lead + support (opt.) & 17 & 0.0371 & 0.4500 & 13/13 & 0.08831 & 0.03531 & 0.08783\\
Proposed & 17 & 0.0252 & 0.3691 & 13/13 & 0.01711 & 0.00790 & 0.01760\\
\bottomrule
\end{tabular}
\par\smallskip
\begin{minipage}{\columnwidth}\scriptsize Maxima are over thirteen equilibria and, for $A,E_g$, their specified bands. Pass requires stability and all frequency-domain limits at the tested equilibria. RMS windows are 2--26, 5--30, and 2--30 s; grid/combined errors use the same-point support reference.\end{minipage}
\end{table}

Fig.~\ref{fig:case_operating_points}(c),(d) shows normalized measurement-to-HESS gains below one at all tested points, satisfying the prescribed gain budget. Peak per-site storage power is 7.846/6.074 MW, within the 20-MW rating; lossless energy excursions are 14.60/2.79 kWh in the single-generator/nine-bus cases. Thus fixed low-order local controllers provide accurate simultaneous isolation and support across the tested conditions within the gain and power limits.

\section{Conclusion}

This paper developed a framework for simultaneous data-center oscillation isolation and grid support. The coupled model identifies a source-rank condition and a pointwise performance conflict under power-only feedback. Additional measurements enable source-selective control when the rank condition holds, while a conditional sensitivity integral constrains isolation across frequency bands. These results guide a convex controller design that coordinates isolation, grid-response tracking, and stability across operating points. In the single-generator and modified nine-bus systems, fixed local controllers satisfied the prescribed requirements at thirteen equilibria. Worst combined-event power-error RMS decreased by 83.7\% and 80.0\% relative to equal-order optimized comparators, demonstrating the benefit of jointly shaping both disturbance responses. Future work will extend the framework to state-of-charge and ramp constraints and large-disturbance operation.

\end{document}